\documentclass{article}

\usepackage[T1]{fontenc}
\usepackage[latin1]{inputenc}
\usepackage{amssymb,amsmath,amsthm,enumerate,amsfonts}
\usepackage{verbatim} 
\usepackage{color} 
\usepackage{pgf,tikz}
\usetikzlibrary{quotes,angles,patterns,snakes}
\usepackage{mathrsfs}
\usepackage[numbers]{natbib}
\usepackage{caption}
\usepackage{subfigure}
\usepackage{graphicx}

\newcounter{examplecounter}
\newenvironment{example}[1]{\begin{quote}%
    \refstepcounter{examplecounter}%
  \textbf{Example \arabic{examplecounter}: #1}%
  \\
}{%
\end{quote}}

\newcounter{exercisecounter}

\usetikzlibrary{arrows}
\definecolor{ttqqqq}{rgb}{0.2,0.,0.}
\definecolor{cqcqcq}{rgb}{0.7529411764705882,0.7529411764705882,0.7529411764705882}
\definecolor{zzttqq}{rgb}{0.6,0.2,0.}
\definecolor{qqqqff}{rgb}{0.,0.,1.}
\definecolor{qqzzqq}{rgb}{0.,0.6,0.}
\definecolor{zzttqq}{rgb}{0.6,0.2,0.}

\newtheorem{theorem}{Theorem}

\newtheorem{lemma}[theorem]{Lemma}
\newtheorem{proposition}[theorem]{Proposition}
\newtheorem{corollary}[theorem]{Corollary}
\newtheorem{definition}[theorem]{Definition}

\author{Teun \textsc{Janssen}, C\'eline \textsc{Swennenhuis}, Abdoul \textsc{Bitar},  \\ 
 Thomas \textsc{Bosman}, Dion \textsc{Gijswijt}, Leo van \textsc{Iersel}, \\
 St\`ephane \textsc{Dauz\`ere-P\'er\`es},
Claude \textsc{Yugma}}
\title{Parallel Machine Scheduling with a Single Resource per Job}
\begin{document}

\maketitle 
\begin{abstract}
We study the problem of scheduling jobs on parallel machines minimizing the total completion time, with each job using exactly one resource. First, we derive fundamental properties of the problem and show that the problem is polynomially solvable if $p_j = 1$. Then we look at a variant of the shortest processing time rule as an approximation algorithm for the problem and show that it gives at least a $(2-\frac{1}{m})$-approximation. Subsequently, we show that, although the complexity of the problem remains open, three related problems are $\mathcal{NP}$-hard. In the first problem, every resource also has a subset of machines on which it can be used. In the second problem, once a resource has been used on a machine it cannot be used on any other machine, hence all jobs using the same resource need to be scheduled on the same machine. In the third problem, every job needs exactly two resources instead of just one. 
\end{abstract}

\section{Introduction}
In this paper, we study a variant of the problem of scheduling jobs on parallel machines with the objective to minimize the total completion time, i.e. the sum of the completion times of all jobs. In this variant, each job uses exactly one resource, thus partitioning the jobs. There is only one unit of each resource available at any time, so jobs using the same resource cannot be processed simultaneously. Using the classification of \citet{graham1979optimization}, we will denote the problem as $P|\textit{partition}|\sum_j C_j$. In this classification, scheduling problems are classified by parameters $\alpha|\beta|\gamma$ where the $\alpha$-field describes the machine environment, the $\beta$-field indicates job characteristics and the $\gamma$-field reflects the optimality criteria.\\

The problem is motivated by a scheduling problem found in the semiconductor industry. The wafer, which contains the chips, will visit different production bays multiple times during its production cycle. The expensive photolithography pieces of equipment are often the bottleneck of the production line. Hence, the overall performance of the factory can be improved by raising the equipment throughput on these tools (which is achieved by minimizing $\sum_j C_j$). Photolithography is a process to transfer the geometric pattern of a chip-design onto a wafer.  This is done by putting light through a reticle onto the production wafer. This reticle contains the geometrical pattern of the  computer chip. Thus, when trying to schedule jobs in the photolithography bay, we need to make sure the reticle (the resource) is available when a job is processed. \\

In the scheduling problem $P||\sum_j C_j$, one wants to minimize the total completion time while scheduling on a set of parallel machines. It is well-known from the literature that $P||\sum_j C_j$ is polynomially solvable by using the shortest processing time first (SPT) order on the earliest available machine (\citet{conway1967}). This rule makes sure that every time a machine finishes a job, it will be assigned, from among the jobs waiting, the job with a shortest processing time. \\

Our problem adds auxiliary resource constraints. \citet{Blaz} describe the resource requirements with the entry $\text{res} \lambda\delta\rho$ in the $\beta$ field of the scheduling problem. The number of different resources is given by $\lambda \in \{ \cdot , c_\lambda\}$. If $\lambda=c_\lambda$, the number of resources is given by $c_\lambda$. If $\lambda=\cdot$ it is part of the input. The resource capacities are denoted by $\delta=\{\cdot,c_\delta \}$. If $\delta=c_\delta$, there is exactly $c_\delta$ of every resource available. If $\delta=\cdot$, the total amount available of a resource is part of the input. The resource requirements per job are denoted by $\rho=\{\cdot,c_\rho \}$. If $\rho=c_\rho$, every job needs exactly $c_\rho$ of a resource it requires. If $\rho=\cdot$, the amount required is part of the input.\\

The type $\text{res}\cdot11$ implies that per resource type, there is one resource available at any given time and this resource will be used entirely if a job needing this resource is processed. This implies that jobs that share the same resource cannot be processed simultaneously. When only $\text{res}\cdot 11$ is in the $\beta$ field of the scheduling problem, a job can need any number of resources. This does not capture that every job in the lithography bay needs only one resource, the reticle. Therefore, we indicate the problem within this paper by \textit{partition} in the $\beta$ field of the scheduling problem. Hence, \textit{partition} is a special case of $\text{res}\cdot11$ resources.\\ 
We know from \citet{Blaz} that if the number of machines is $m\geq 3$ and there are no further restrictions on the $P|\text{res}\cdot11, p_j=1|\sum_j C_j$, the problem is $\mathcal{NP}$-hard. The proof is based on a reduction from partition into triangles and uses multiple resources per job. It is proven by \citet{garey1975complexity} that $P2|\text{res}\cdot 11, p_j=1|\sum_j C_j$ can be solved in polynomial time by a reduction to the matching problem.\\ 

The problem can also be viewed as a special case of $PD|\text{res } 1 \cdot 1 |\sum_j C_j$. In $PD|\text{res } 1 \cdot 1 |\sum_j C_j$, we have dedicated machines and are given only one resource of a certain quantity $c_\delta$ and every job needs exactly $1$ from that resource. We can rewrite $P|\textit{partition}|\sum_j C_j$ to $PD|\text{res } 1 \cdot 1 |\sum_j C_j$ by taking $c_\delta$ equal to the number of machines and introducing a dedicated machine for every resource, such that all jobs that share a resource have to be processed on the same machine. 
One could also view $P|\textit{partition}|\sum_j C_j$ as a special version of scheduling with conflicts. In scheduling with conflicts, we have again parallel machines, the total completion time objective and jobs cannot be processed at the same time if they share an edge in the conflict graph $G=(J,E)$ with an edge between two jobs if they share the same resource. In our problem, $G$ is a collection of cliques.  \\

Our contribution is as follows. First we prove that allowing preemptions to the problem, does not change the problem. Using that, we conclude that in each optimal schedule for $P|\textit{partition}|\sum_j C_j$, all jobs sharing the same resource must be processed in order of processing time. Restricting the problem to $p_j=1$ is proven to be polynomially solvable. Thereafter we look at an approximation algorithm for $P|\textit{partition}|\sum_j C_j$, based on a variant of the shortest processing time (SPT) rule that takes the partition constraints into account. We prove that it gives a $(2-\frac{1}{m})$-approximation and show that it cannot give an $\alpha$-approximation with $\alpha < \frac{4}{3}$. In the last three sections we look at three related problems and show that they are $\mathcal{NP}$-hard. The first problem has additional processing set restrictions for resources, meaning each resource has a set $\mathcal{M}_r$ of machines on which it can be used. From this, we can also conclude that the problem with unrelated machines, i.e. $R|\textit{partition}|\sum_j C_j$, is also $\mathcal{NP}$-hard. This is the situation of the scheduling of Photolithography machines in practice. The second related problem assumes that resources are unmovable, meaning that once a resource is used on a machine, it can thereafter only be used on that specific machine. In the last related problem, each job has at most $q$ resources with $q\ge2$ a constant.\\

\section{Definitions}
Before we begin our analysis of the problem and its solutions, we first formally define \emph{partition} as part of the $\beta$ field and some intermediary concepts. We have $m$ identical machines and let $J$ be the set of jobs that are to be scheduled. Each $j \in J$ has a \emph{processing time} $p_j$ and each machine can only process a single job at a time. We will denote $C_j$ as the \emph{completion time} of job $j$ in a feasible schedule for an instance. We want to minimize the sum of the completion times (total completion time). 

\begin{definition}
If \emph{partition} is in the $\beta$ field, there is a partition $R$ of $J$, i.e., there is a collection of subsets $R = \{r^1,\ldots,r^s\}$ with $r^k \subseteq J$, where every job is contained in exactly one of the subsets. If $j,j' \in r^k$, $j$ and $j'$ cannot be processed at the same time. Furthermore, we want to define which resource is used by which job. Let $r_j=\{ r^k \in R \ | \ j \in r^k \}$, i.e., all subsets that contain job $j$. If two jobs share the same resource, we will denote this by $r_j=r_{j'}$, which implies that $r_j \cap r_{j'} \neq \emptyset $.\\
\end{definition}

When we look at a job, we will often consider the other jobs that share the same resource. We will, therefore, introduce the concept of slack, which, intuitively, is the amount of time before and after the job that its resource is not being used.

\begin{definition}
A job $j$ has \emph{positive slack} $d^{+}\geq 0$, which is the largest non negative number, such that in a given schedule all jobs $j'\in J$ which have the same resource and start after job $j$, start at least $d^{+}$ time units after $j$ finishes. More formally, we define $d^{+}$ as
\[d^{+} := \min\left\{ C_{j'}-p_{j'}-C_j | j' \in J \text{ satisfies } r_{j'} = r_j \text{ and } C_{j'} > C_j \right\} \]
where we define $+\infty$ as the minimum over the empty set. \\ 
Similarly, a job $j$ has \emph{negative slack} $d^{-}\leq 0$, which is the largest non negative number, such that all jobs $j'\in J$ which have the same resource and start before job $j$, finish at least $d^{-}$ time units before $j$ starts.  More formally, 
\[d^{-} := \min\left\{ C_{j}-p_{j}-C_{j'} | j' \in J \text{ satisfies } r_{j'} = r_j \text{ and } C_{j'} < C_j \right\} \]
where we define $+\infty$ as the minimum over the empty set. \\
The \emph{slack} $d> 0$ of a job $j$, we define as $d=\min\{d^{+},d^{-}\}$.
\end{definition}

We defined the slack of a job by considering all jobs that share the same resource, but often we are only interested in the last job before and the first job after a job $j$ that use the same resource. We therefore introduce the following concept.

\begin{definition}
Let $d^{+}$ be the positive slack of $j$, we call a job pair $(j,j')$ a \emph{blocking pair} if  $r_j = r_{j'}$ and $C_j = C_{j'}-p_{j'}-d^{+}$. Thus, $j'$ is the first job to start after $C_j$ that uses the same resource as job $j$. A blocking pair $(j,j')$ is \emph{tight} if $d^{+}=0$.
\end{definition}

Given a tight pair where the two jobs are not on the same machine, it can be advantageous to construct a schedule were they actually are on the same machine. We will call this operation `untangling'. Before we define it properly however we need to first define a job's suffix. 

\begin{definition}
Let job $j$ be processed on machine $i$. The suffix of job $j$, $\mathcal{S}(j)$, is the set of jobs $j' \in J$ that are also processed on machine $i$ with $C_{j'}\ge C_j$ (excluding job $j$). 
\end{definition}

\begin{definition}
Given a tight pair $(j,j')$, let $i$ be the machine on which $j$ is processed and $i'$ be the machine on which $j'$ is processed. \emph{Untangling} $(j,j')$ is the operation that changes the schedule by swapping suffices between the machines $i$ and $i'$, i.e., we move $j'$ and $S(j')$ to machine $i$ and $S(j)$ to machine $i'$. 
\end{definition}
Since we work with parallel machines, untangling will not change any of the start or completion times of the jobs. Hence, it will not create any resource conflicts and the objective function remains the same.\\ 

\section{Problem Properties}
In this section, we consider the structure of optimal solutions. We show that there is no non-trivial idle time in an optimal solution and that given a resource, the jobs using that resource are scheduled from shortest to longest. We will continue by looking at the complexity of the problem. Whether or not $P|\textit{partition}|\sum_j C_j$ is $\mathcal{NP}$-hard remains an open problem, but we can show that when $p_j=1$ the problem is polynomially solvable. We also show that the problem with preemptions is equal to the problem without preemptions.\\  

We first note that if $|R|<m$ the problem becomes trivial. In that case, one will put all jobs which use the same resource in shortest processing time order on one machine. 
We continue by looking at idle times in a solution. 

\begin{lemma} \label{idletime}
For every instance of $P|\textit{partition}|\sum_j C_j$ there exists an optimal solution that contains no idle times.
\end{lemma}
\begin{proof}
Suppose that, in an optimal schedule  with idle times. We begin by untangling all tight pairs. If an idle time remains, we consider the last idle time, which appears on machine $i$ that starts on time $t_1$ and ends at time $t_2$. Since we have untangled all tight pairs, the resource used by the job on machine $i$ starting at time $t_2$ is not used until time $t_2$. Hence, we can start the processing of this job earlier. We can then schedule the job either at time $t_1$ or at the last time before $t_2$ its resource was used. This would reduce the completion time of this job. Therefore, after untangling, there cannot be any idle times. Since untangling does not change completion times there is an optimal schedule without idle times. 
\end{proof}
In the proof, we saw it is easy to turn an arbitrary optimal schedule into a schedule where tight pairs $(j,j')$ are processed on the same machine. We call such a schedule a \emph{tight} schedule.
\begin{definition}
A \emph{tight schedule} is a schedule without any idle time and in which each tight blocking pair $(j,j')$ is executed on the same machine. 
\end{definition}
Notice that untangling results in jobs using the same resource being processed one after another on the same machine. We will call these job sequences \emph{trains}. 
\begin{definition}
A \emph{train sequence} $T(j_1)$ in a schedule is a maximal sequence of consecutively jobs $j_1,j_2,...,j_c$ on the same machine using the same resource.
\end{definition}
Notice that a tight schedule only consists of train sequences $T(j_k)$ with nonzero slack between the train sequences of the same resource, where the $j_k$ are the first jobs to be scheduled when a machine changes resource. \\  

We continue by looking at the order in which jobs that use the same resource are processed. We will prove that this is from shortest to longest processing time. We will prove this by first looking at the problem with preemptions, notated by $P|\textit{partition}, \textit{prmp}|\sum_j C_j$. In a preemptive schedule, the total amount of processing done on the job needs to be equal to its processing time ($p_j$), but jobs can be interrupted at any time and the processing done is not lost. A job can thus be split into multiple parts, possibly processed on different machines. We begin by defining these more precisely. 
\begin{definition}
A \emph{job part} $j^l$ is the $l$th maximal part of the job $j$ that is processed without interruption on a single machine with positive length. The superscript $l$ will be omitted when it is of no importance. A pair of job parts $(j,j')$ is called a \emph{blocking pair} if  $r_j = r_{j'}$ and $j'$ is the first job part to start after $j$ that uses the same resource as job part $j$. A blocking pair $(j,j')$ is \emph{tight} if job part $j'$ starts at the time $j$ ends, i.e., $d^{+}=0$.
\end{definition}
The definitions for blocking pairs, slack, suffix, train sequences, tight schedules and untangling can easily be extended to the case of job parts.  

\begin{lemma} \label{SPTorderPRMP}
In an optimal schedule for $P|\textit{partition}, \textit{prmp}|\sum_j C_j$  all jobs sharing the same resource must be processed in SPT-order, i.e., if job $j$ and $j'$ both use resource $r \in R$ and $p_j < p_{j'}$ then $C_j < C_{j'}$. Furthermore, if $C_j < C_{j'}$, all job parts of $j$ will be processed before any job parts of $j'$. 
\end{lemma}
\begin{proof}

Suppose we have an optimal schedule $S$ where there is not the case. Then there is a resource $r\in R$ and two jobs using this resource ($r_j = r_{j'}$), job $j$ and job $j'$, with $C_j < C_{j'}$ and $p_j>p_{j'}$. From $S$ we get an ordering of the jobs using resource $r$. Let $j_{S(r,p)}$ denote the $l$th job finishing in $S$ using resource $r$. \\ 
Create a new schedule $S'$ which is identical to $S$ except for all jobs using resource $r$. We remove from $S$ all job parts using resource $r$. This will remove $t=\sum_{j\in J | r_j =r} p_j$ units of processing from the schedule. We fill these units of processing again with the jobs using resource $r$ but now we process them in an SPT-order. We processes the first job in the ordering $j_{S'(r,1)}$ in the first $p_{j_{S'(r,1)}}$ units of $t$. We schedule the second job in the ordering $j_{S'(r,2)}$ in the first $p_{j_{S'(r,2)}}$ units of $t$ after $C_{j_{S'(r,1)}}$ and so on until all jobs using resource $r$ are scheduled. The resource $r$ will be used in the same time as in $S$ by only a single job, hence $S'$ is a feasible schedule. Furthermore, it holds that $C_{j_{S(r,1)}} \leq C_{j_{S(r,p)}} \forall p$, since 
\begin{equation}
p_{j_{S'(r,1)}}+p_{j_{S'(r,2)}}+\ldots+ p_{j_{S'(r,p)}} \leq p_{j_{S(r,1)}}+p_{j_{S(r,2)}}+\ldots+ p_{j_{S(r,p)}}. \label{eq1_SPTorderPRMP}
\end{equation}
Since $p_j>p_{j'}$, equation (\ref{eq1_SPTorderPRMP}) is satisfied with inequality for the $l$ job and $S$ cannot be optimal. 
\end{proof}

\begin{theorem} \label{nopreempt}
There is an optimal schedule for $P | \textit{partition}, \textit{prmp}| \sum_j C_j$ without any preemptions.
\end{theorem}
\begin{proof}
Assume not, then take any optimal tight schedule with a minimal amount of preemptions. Let $t_1$ be the time of the last occurring preemption, let job $j$ be the job that is being interrupted with resource $r$ on machine $i$, let $j^l$ be the respective job part. Let $t_2$ be the time that job part $j^{l+1}$ starts on machine $i'$. Let $j'$ be the job part on machine $i$ that starts at $t_1$ and let $r'$ be its resource.\\
We know that $t_2>t_1$, otherwise we would not have a tight schedule. Furthermore, following from Lemma \ref{SPTorderPRMP}, resource $r$ cannot be used by another job between $t_1$ and $t_2$. \\
We also know that  $i \neq i'$ by the following argument illustrated in Figure \ref{i=i'}. Assume $i = i'$. Take $\epsilon >0$ as the minimal negative slack of all train sequences between $t_1$ and $t_2$ on machine $i$. Then, one can move all jobs between $t_1$ and $t_2$ on machine $i$ $\epsilon$ to the front and then split $j^l$ on $t_1 - \epsilon$ and move the second part from $[t_1-\epsilon,t_1]$ to $[t_2- \epsilon,t_2]$. Since there is no preemption after $t_1$, at least one job finishes earlier in this new situation and no jobs finishing later. Thus, the original schedule was not optimal.\\

\begin{figure}[h!]
\centering
\begin{tikzpicture}[scale = 0.8]
\fill[fill =gray, draw = black] (0,3) rectangle (4,4);
\fill[fill = white, draw = black] (4,3) rectangle (8,4);
\fill[fill =gray, draw = black] (8,3) rectangle (12,4);
\node at (2,3.5) {$j^l$};
\node at (10,3.5) {$j^{l+1}$};
\node at (4,4.25){$t_1$};
\node at (8,4.25){$t_2$};
\node at (6,3.5) {$\cdots$};
\draw[snake=snake] (4,4) -- (4,3);
\draw[->] (6,2.5) -- (6,1.5);
\fill[fill =gray, draw = black] (0,0) rectangle (3,1);
\fill[fill = white, draw = black] (3,0) rectangle (7,1);
\fill[fill =gray, draw = black] (7,0) rectangle (12,1);
\node at (2,0.5) {$j^l$};
\node at (10,0.5) {$j^{l+1}$};
\node at (3,1.25){$t_1 - \epsilon$};
\node at (7,1.25){$t_2 - \epsilon$};
\node at (5,0.5) {$\cdots$};
\draw[snake=snake] (3,1) -- (3,0);
\draw[dashed] (7,4) -- (7,3);
\draw[dashed] (3,4) -- (3,3);
\draw[->] (7.5,4.20) -- (4.5,4.20);
\node at (6,4.5) {move $\epsilon$ to front};
\end{tikzpicture}
\caption{Situation where $i=i'$. The squiggly line represents a preemption.}
\label{i=i'}
\end{figure}
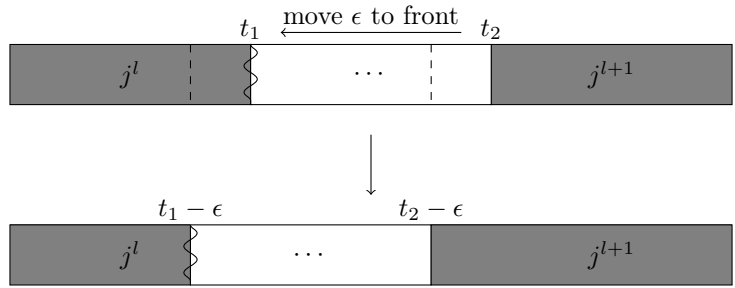

We know that job $j'$ cannot end before or on $t_2$. As illustrated in Figure \ref{Cj'<t2}, if it would, one could move job $j'$ $\epsilon>0$ to the front, where $\epsilon$ is the negative slack of job $j'$. This would split the job part $j^l$ on $t_1 - \epsilon$ and moving the part that was executed during the interval $[t_1-\epsilon,t_1]$ to the back of $j'$. This leads to a feasible schedule, since we defined $t_2$ as the time that job part $j^{l+1}$ starts and resource $r$ is not used between $t_1$ and $t_2$. Furthermore, since $j'$ is a finishing job part and it finishes $\epsilon$ earlier, this will lead to a schedule with better objective value.\\

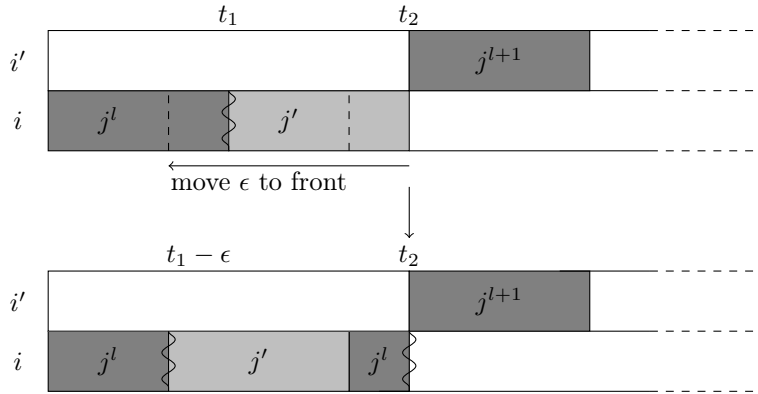
\begin{figure}[h!]
\centering
\begin{tikzpicture}[scale = 0.8]

\fill[fill =gray, draw = black] (0,4) rectangle (3,5);%
\fill[fill =lightgray, draw = black] (3,4) rectangle (6,5);
\draw (6,4) -- (10,4);
\draw (6,5) -- (10,5);
\draw[dashed] (10,4) -- (12,4);
\draw[dashed] (10,5) -- (12,5);
\fill[fill = white, draw = black] (0,5) rectangle (6,6);
\fill[fill =gray, draw = black] (6,5) rectangle (9,6);
\draw (9,6) -- (10,6);
\draw[dashed] (10,6) -- (12,6);
\node at (1,4.5) {$j^l$};
\node at (7.5,5.5) {$j^{l+1}$};
\node at (3,6.25){$t_1$};
\node at (6,6.25){$t_2$};
\node at (4,4.5) {$j'$};
\node at (-0.5,4.5) {$i$};
\node at (-0.5,5.5) {$i'$};
\draw[snake=snake] (3,5) -- (3,4);
\draw[dashed] (5,4) -- (5,5);
\draw[dashed] (2,4) -- (2,5);
\draw[->] (6,3.4) -- (6,2.55);
\draw[->] (6,3.75) -- (2,3.75);
\node at (3.5,3.5) {move $\epsilon$ to front};

\fill[fill =gray, draw = black] (0,0) rectangle (2,1);%
\fill[fill =gray, draw = black] (5,0) rectangle (6,1);%
\fill[fill =lightgray, draw = black] (2,0) rectangle (5,1);
\draw (5.5,0) -- (10,0);
\draw (6.5,1) -- (10,1);
\draw[dashed] (10,0) -- (12,0);
\draw[dashed] (10,1) -- (12,1);
\draw[dashed] (10,2) -- (12,2);
\fill[fill = white, draw = black] (0,1) rectangle (6,2);
\fill[fill =gray, draw = black] (6,1) rectangle (9,2);
\draw (8.5,2) -- (10,2);
\node at (1,0.5) {$j^l$};
\node at (5.5,0.5) {$j^l$};
\node at (7.5,1.5) {$j^{l+1}$};
\node at (2.5,2.25){$t_1-\epsilon$};
\node at (6,2.25){$t_2$};
\node at (3.5,0.5) {$j'$};
\node at (-0.5,0.5) {$i$};
\node at (-0.5,1.5) {$i'$};
\draw[snake=snake] (2,1) -- (2,0);
\draw[snake=snake] (6,1) -- (6,0);
\end{tikzpicture}
\caption{Finding a better solution if $C_{j'} \leq t_2$}
\label{Cj'<t2}
\end{figure}

As a result, there is only one possible situation that can occur with the last preemption: The last preemption is at a different machine than where it later continues and job $j'$ starts at $t_1$ on machine $i$ and does not finish before or on $t_2$. The partial schedule is shown in Figure \ref{situation}.
 
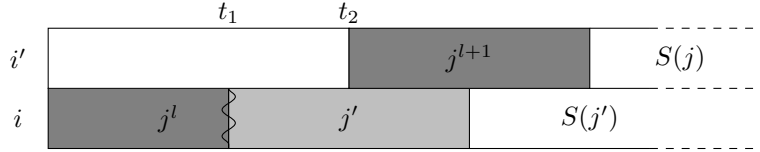
\begin{figure}[h!]
\centering
\begin{tikzpicture}[scale = 0.8]

\fill[fill =gray, draw = black] (0,0) rectangle (3,1);
\fill[fill =lightgray, draw = black] (3,0) rectangle (7,1);
\draw (7,0) -- (10,0);
\draw (7,1) -- (10,1);
\draw[dashed] (10,0) -- (12,0);
\draw[dashed] (10,1) -- (12,1);
\fill[fill = white, draw = black] (0,1) rectangle (5,2);
\fill[fill =gray, draw = black] (5,1) rectangle (9,2);
\draw (9,2) -- (10,2);
\draw[dashed] (10,2) -- (12,2);
\node at (2,0.5) {$j^l$};
\node at (7,1.5) {$j^{l+1}$};
\node at (3,2.25){$t_1$};
\node at (5,2.25){$t_2$};
\node at (5,0.5) {$j'$};
\node at (-0.5,0.5) {$i$};
\node at (-0.5,1.5) {$i'$};
\node at (9,0.5) {$S(j')$};
\node at (10.5,1.5) {$S(j)$};
\draw[snake=snake] (3,1) -- (3,0);
\end{tikzpicture}
\caption{Only possible situation in an optimal schedule with preemption.}
\label{situation}
\end{figure}

Let $S(j)$ be job part $j^{l+1}$ on machine $i'$ and its suffix and let $S(j')$ be job part $j'$ on machine $i$ and its suffix. The jobs in these sets all finish, as $j^l$ is the last preemption. When we look at the number of jobs in both sets, there are two possibilities: 

\begin{itemize}
\item $|S(j)|<|S(j')|$. Figure \ref{Sj<Sj'} illustrates this case. Take a maximal $\epsilon>0$ such that all train sequences in $S(j)$ can start $\epsilon$ later (i.e. have $d^{+}\geq \epsilon$) and all train sequences in $S(j')$ can start $\epsilon$ earlier (i.e. have $d^{-}\geq \epsilon$), while staying in a feasible schedule. Move the sets in the mentioned directions and move the interval $[t_1-\epsilon,t_1]$  of job $j$ on machine $i$ to machine $i'$ on the interval $[t_2,t_2+\epsilon]$. Also, move $j^{l+1}$ $\epsilon$ to the back and $j'$ $\epsilon$ to the front. Clearly, this is a feasible schedule. All job parts in the sets $S(j)$ and $S(j')$ are no preemptions, thus the objective value changes by $\epsilon (|S(j)| - |S(j')|)$, and therefore becomes smaller. Hence, this situation cannot happen in an optimal solution. 

\begin{figure}[h!]
\centering
\begin{tikzpicture}[scale = 0.8]

\fill[fill =gray, draw = black] (0,4) rectangle (3,5);
\fill[fill =lightgray, draw = black] (3,4) rectangle (7,5);
\draw (7,4) -- (10,4);
\draw (7,5) -- (10,5);
\draw[dashed] (10,4) -- (12,4);
\draw[dashed] (10,5) -- (12,5);
\fill[fill = white, draw = black] (0,5) rectangle (5,6);
\fill[fill =gray, draw = black] (5,5) rectangle (9,6);
\draw (9,6) -- (10,6);
\draw[dashed] (10,6) -- (12,6);
\node at (2,4.5) {$j^l$};
\node at (7,5.5) {$j^{l+1}$};
\node at (3,6.25){$t_1$};
\node at (5,6.25){$t_2$};
\node at (5,4.5) {$j'$};
\node at (-0.5,4.5) {$i$};
\node at (-0.5,5.5) {$i'$};
\draw[snake=snake] (3,5) -- (3,4);
\draw[->] (5.5,6.25) -- (10,6.25);
\node at (8,6.5) {move $\epsilon$ to back};
\draw[->] (9,3.75) -- (3,3.75);
\node at (5.5,3.5) {move $\epsilon$ to front};
\draw[dashed] (2.5,4) -- (2.5,5);
\draw[dashed] (9.5,5) -- (9.5,6);
\draw[->] (6,3.2) -- (6,2.5);

\fill[fill =gray, draw = black] (0,0) rectangle (2.5,1);
\fill[fill =lightgray, draw = black] (2.5,0) rectangle (6.5,1);
\draw (6.5,0) -- (10,0);
\draw (6.5,1) -- (10,1);
\draw[dashed] (10,0) -- (12,0);
\draw[dashed] (10,1) -- (12,1);
\fill[fill = white, draw = black] (0,1) rectangle (5,2);
\fill[fill =gray, draw = black] (5,1) rectangle (9.5,2);
\draw (8.5,2) -- (10,2);
\draw[dashed] (6.5,4) -- (6.5,5);
\draw[dashed] (10,2) -- (12,2);
\node at (2,0.5) {$j^l$};
\node at (7,1.5) {$j^{l+1}$};
\node at (2.5,2.25){$t_1-\epsilon$};
\node at (5,2.25){$t_2$};
\node at (5.5,0.5) {$j'$};
\node at (-0.5,0.5) {$i$};
\node at (-0.5,1.5) {$i'$};
\draw[snake=snake] (2.5,1) -- (2.5,0);

\end{tikzpicture}
\caption{Finding a better solution if $|S(j)|<|S(j')|$}
\label{Sj<Sj'}
\end{figure}
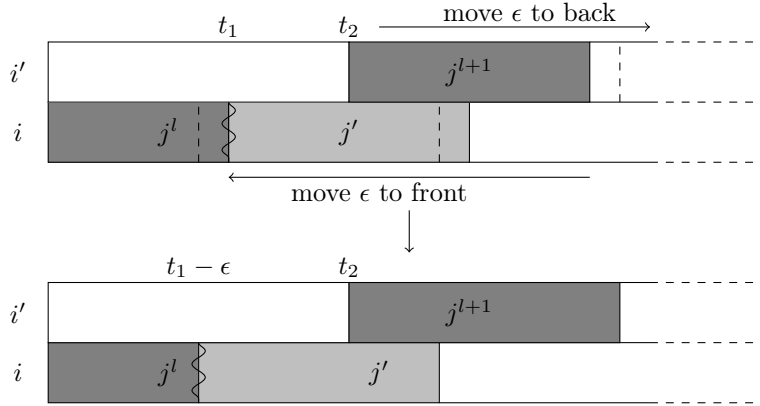

\item $|S(j)|\geq|S(j')|$. Figure \ref{Sj>Sj'} illustrates this case. Take a maximal $\epsilon>0$ such that all train sequences in $S(j)$ can start $\epsilon$ earlier (i.e. have $d^{-}\geq \epsilon$) and all train sequences in $S(j')$ can start $\epsilon$ later (i.e. have $d^{+}\geq \epsilon$), while staying a feasible schedule. Move the sets in the mentioned directions and move the interval $[t_2,t_2+\epsilon]$  of job $j$ on machine $i'$ to machine $i$ on the interval $[t_1,t_1+\epsilon]$. Also move job $j^{l+1}$ to the front and $j'$ to the back. Clearly, this is a feasible schedule. All job parts in the sets $S(j)$ and $S(j')$ are no preemptions, thus the objective value changes by $\epsilon (|S(j')| - |S(j)|)$. Hence, $|S(j)| > |S(j')|$ cannot happen in an optimal solution. \\

\begin{figure}[h!]
\centering
\begin{tikzpicture}[scale = 0.8]

\fill[fill =gray, draw = black] (0,4) rectangle (3,5);
\fill[fill =lightgray, draw = black] (3,4) rectangle (7,5);
\draw (7,4) -- (10,4);
\draw (7,5) -- (10,5);
\draw[dashed] (10,4) -- (12,4);
\draw[dashed] (10,5) -- (12,5);
\fill[fill = white, draw = black] (0,5) rectangle (5,6);
\fill[fill =gray, draw = black] (5,5) rectangle (9,6);
\draw (9,6) -- (10,6);
\draw[dashed] (10,6) -- (12,6);
\node at (2,4.5) {$j^l$};
\node at (7,5.5) {$j^{l+1}$};
\node at (3,6.25){$t_1$};
\node at (5,6.25){$t_2$};
\node at (5,4.5) {$j'$};
\node at (-0.5,4.5) {$i$};
\node at (-0.5,5.5) {$i'$};
\draw[snake=snake] (3,5) -- (3,4);
\draw[->] (10,6.25) -- (5.5,6.25);
\node at (7.5,6.5) {move $\epsilon$ to front};
\draw[->] (3,3.75) -- (7.5,3.75);
\node at (5.5,3.5) {move $\epsilon$ to back};
\draw[dashed] (7.5,4) -- (7.5,5);
\draw[dashed] (8.5,5) -- (8.5,6);
\draw[->] (6,3.2) -- (6,2.5);

\fill[fill =gray, draw = black] (0,0) rectangle (3.5,1);
\fill[fill =lightgray, draw = black] (3.5,0) rectangle (7.5,1);
\draw (7.5,0) -- (10,0);
\draw (7.5,1) -- (10,1);
\draw[dashed] (10,0) -- (12,0);
\draw[dashed] (10,1) -- (12,1);
\fill[fill = white, draw = black] (0,1) rectangle (5,2);
\fill[fill =gray, draw = black] (5,1) rectangle (8.5,2);
\draw (8.5,2) -- (10,2);

\draw[dashed] (10,2) -- (12,2);
\node at (2,0.5) {$j^l$};
\node at (7,1.5) {$j^{l+1}$};
\node at (3.5,2.25){$t_1+\epsilon$};
\node at (5,2.25){$t_2$};
\node at (5.5,0.5) {$j'$};
\node at (-0.5,0.5) {$i$};
\node at (-0.5,1.5) {$i'$};
\draw[snake=snake] (3.5,1) -- (3.5,0);
\draw[dashed] (3.5,5) -- (3.5,4);
\end{tikzpicture}
\caption{Finding a better solution if $|S(j)|>|S(j')|$}
\label{Sj>Sj'}
\end{figure}
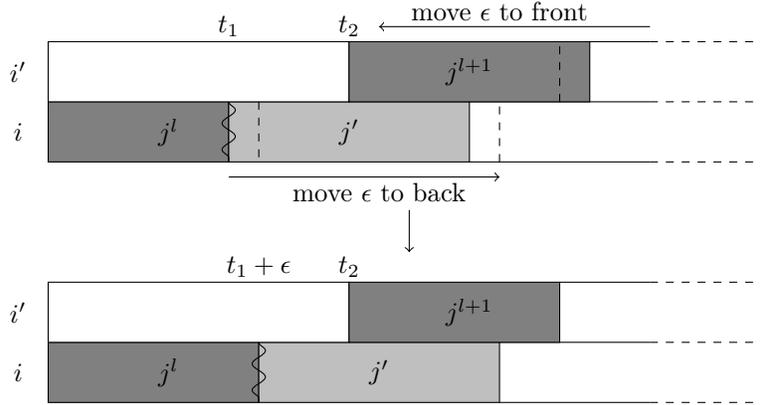
\end{itemize}
Therefore, in an optimal solution, $|S(j)| = |S(j')|$. Since the $\epsilon$ was chosen maximal in the $|S(j)| \geq |S(j')|$ case, there must be a new tight pair, created by moving the suffices backward and forward. Untangle new tight pairs, such that the schedule becomes tight again. In the new schedule, it must hold that again $|S(j)|=|S(j')|$, because otherwise the solution was not optimal. It is possible that (one of the) newly tight pairs was $(j^l,j^{l+1})$, in that case a preemption was removed contradicting the assumption that the number of preemptions was minimal. If not, keep repeating moving the suffices and job parts as described and untangling. This can be done only a finite number of times since there is a maximum number of tight pairs (bounded by $n$) and because during this process, tight pairs remain tight and no new job parts are created (and with it new preemptions). Hence, there is an optimal schedule containing no preemptions.
\end{proof}

Following from Theorem \ref{nopreempt} and Lemma \ref{SPTorderPRMP}, we know that $P|\textit{partition},\textit{prmp} |\sum_j C_j$ and $P|\textit{partition}|\sum_j C_j$ have the same objective value and we obtain the following result.

\begin{theorem}\label{SPTorderOPT} 
In an optimal schedule for $P|\textit{partition}|\sum_j C_j$, all jobs sharing the same resource are processed in SPT-order, i.e., if job $j$ and $j'$ both use resource $r \in R$ and $p_j < p_{j'}$ then $C_j < C_{j'}$.  
\end{theorem}

We continue by looking at $p_j=1$, since $P3|\text{res}.11, p_j=1|\sum_j C_j$ is $\mathcal{NP}$-hard. Surprisingly, with at most one resource per job, the problem becomes polynomially solvable.

\begin{theorem}\label{res_p1}
$P|\textit{partition}, p_j=1|\sum_j C_j$ is polynomially solvable.
\end{theorem}
\begin{proof}
The problem can be reduced to an instance of the Min-Cost Flow problem. Next to the source node $s$ and the target node $t$, we construct three sets of nodes $V_{J}$, $V_{\text{res},\text{pos}}$, $V'_{\text{res},\text{pos}}$ and $V_{M,\text{pos}}$. The set $V_{J}$ corresponds to the jobs. It has $n$ nodes; one for every job. The set $V_{\text{res},\text{pos}}$ corresponds to resource needed and the completion time/position of the job on a machine in the schedule. Completion time and position on a machine are in this case equal since $p_j=1$ and we may assume no idle times from Lemma \ref{idletime}. The set $V_{\text{res},\text{pos}}$ has $n |R|$ nodes. The set $V'_{\text{res},\text{pos}}$ is a duplicate of these nodes. The set $V_{\text{M},\text{pos}}$ corresponds to machine used and the completion time/position of the job on the machine in the schedule. It has $m n$ nodes.\\
We start by constructing arcs with cost 0 and capacity 1. The first set of arcs is constructed from $s$ to every node in $V_{J}$. From every node $v_j \in V_{J}$, we construct an arc to every node in $V_{\text{res},\text{pos}}$ that corresponds to the resource needed by the job $j$. We construct from every node $v_{r,p} \in V_{\text{res},\text{pos}}$ an arc to the corresponding node with the same resource and position $v'_{r,p} \in V'_{\text{res},\text{pos}}$. Next, we construct from every node $v'_{r,p} \in V'_{\text{res},\text{pos}}$ an arc to every node in $V_{M,\text{pos}}$ having the same position $p$. Lastly, we construct arcs with capacity $1$ and cost equal to position $p$ from every node $v_{m,p} \in V_{M,\text{pos}}$ to $t$. We thus have $n(1+n+|R|+|R|m+m)$ arcs. Lastly, we require that we have at least $n$ units of flow from $s$ to $t$. 

\begin{figure}[ht!]
\centering
\begin{tikzpicture}[line cap=round,line join=round,>=triangle 45,x=0.78cm,y=0.78cm]
\clip(-6.8,-4.5) rectangle (9,5.9);
\draw [line width=0.8pt] (-6.,0.)-- (-3.,1.);
\draw [line width=0.8pt] (-6.,0.)-- (-3.,-1.);
\draw [line width=0.8pt] (-3.,1.)-- (0.,4.);
\draw [line width=0.8pt] (-3.,1.)-- (0.,3.);
\draw [line width=0.8pt] (-3.,1.)-- (0.,2.);
\draw [line width=0.8pt] (-3.,1.)-- (0.,1.);
\draw [line width=0.8pt] (-3.,-1.)-- (0.,-1.);
\draw [line width=0.8pt] (-3.,-1.)-- (0.,-2.);
\draw [line width=0.8pt] (-3.,-1.)-- (0.,-3.);
\draw [line width=0.8pt] (-3.,-1.)-- (0.,-4.);
\draw [line width=0.8pt] (-6.,0.)-- (-3.,2.);
\draw [line width=0.8pt] (-6.,0.)-- (-3.,-2.);
\draw [line width=0.8pt] (4.2,4.)-- (8.,0.);
\draw [line width=0.8pt] (4.2,3.)-- (8.,0.);
\draw [line width=0.8pt] (4.2,2.)-- (8.,0.);
\draw [line width=0.8pt] (4.2,1.)-- (8.,0.);
\draw [line width=0.8pt] (4.2,-1.)-- (8.,0.);
\draw [line width=0.8pt] (4.2,-2.)-- (8.,0.);
\draw [line width=0.8pt] (4.2,-3.)-- (8.,0.);
\draw [line width=0.8pt] (4.2,-4.)-- (8.,0.);
\draw [line width=0.8pt] (-3.,2.)-- (0.,4.);
\draw [line width=0.8pt] (-3.,2.)-- (0.,3.);
\draw [line width=0.8pt] (-3.,2.)-- (0.,2.);
\draw [line width=0.8pt] (-3.,2.)-- (0.,1.);
\draw [line width=0.8pt] (-3.,-2.)-- (0.,-4.);
\draw [line width=0.8pt] (-3.,-2.)-- (0.,-3.);
\draw [line width=0.8pt] (-3.,-2.)-- (0.,-2.);
\draw [line width=0.8pt] (-3.,-2.)-- (0.,-1.);
\draw (-3.6,6) node[anchor=north west] {Jobs};
\draw (-0.5,6) node[anchor=north west] {\parbox{2.256010926485105 cm}{Resource \&  \\ Position}};
\draw (3.3,6) node[anchor=north west] {\parbox{2.0276453034878132 cm}{Machine \& \\ Position}};
\draw [line width=0.8pt] (0.,4.007796456697401)-- (1.2,4.);
\draw [line width=0.8pt] (0.,3.)-- (1.2,3.);
\draw [line width=0.8pt] (0.,2.)-- (1.2,2.);
\draw [line width=0.8pt] (0.,1.)-- (1.2,1.);
\draw [line width=0.8pt] (0.,-1.)-- (1.2,-1.);
\draw [line width=0.8pt] (0.,-2.)-- (1.2,-2.);
\draw [line width=0.8pt] (0.,-3.)-- (1.2,-3.);
\draw [line width=0.8pt] (0.,-4.)-- (1.2,-4.);
\draw [line width=0.8pt] (1.2,4.)-- (4.2,4.);
\draw [line width=0.8pt] (1.2,3.)-- (4.2,3.);
\draw [line width=0.8pt] (1.2,3.)-- (4.2,-2.);
\draw [line width=0.8pt] (1.2,4.)-- (4.2,-1.);
\draw [line width=0.8pt] (1.2,2.)-- (4.2,1.9905667868879942);
\draw [line width=0.8pt] (1.2,2.)-- (4.2,-3.);
\draw [line width=0.8pt] (1.2,1.)-- (4.2,1.);
\draw [line width=0.8pt] (1.2,1.)-- (4.2,-4.);
\draw [line width=0.8pt] (1.2,-1.)-- (4.2,-1.);
\draw [line width=0.8pt] (1.2,-1.)-- (4.2,4.);
\draw [line width=0.8pt] (1.2,-2.)-- (4.2,-2.);
\draw [line width=0.8pt] (1.2,-2.)-- (4.2,3.);
\draw [line width=0.8pt] (1.2,-3.)-- (4.2,-3.);
\draw [line width=0.8pt] (1.2,-3.)-- (4.2,1.9905667868879942);
\draw [line width=0.8pt] (1.2,-4.)-- (4.2,-4.);
\draw [line width=0.8pt] (1.2,-4.)-- (4.2,1.);
\begin{scriptsize}
\draw [fill=black] (-6.,0.) circle (2.5pt);
\draw[color=black] (-6.2,0.3444312544491866) node {$s$};
\draw [fill=black] (-3.,2.) circle (2.5pt);
\draw[color=black] (-3.05,2.4) node {$J_1$};
\draw [fill=black] (-3.,-2.) circle (2.5pt);
\draw[color=black] (-3.05,-2.4) node {$J_4$};
\draw [fill=black] (4.2,4.) circle (2.5pt);
\draw[color=black] (4.3,4.3) node {$M_1,1$};
\draw [fill=black] (4.2,-4.) circle (2.5pt);
\draw[color=black] (4.3,-4.33) node {$M_2,4$};
\draw [fill=black] (8.013895046403343,0.) circle (2.5pt);
\draw[color=black] (8.261291137983743,0.3444312544491866) node {$t$};
\draw [fill=black] (-3.,1.) circle (2.5pt);
\draw[color=black] (-3.05,1.4) node {$J_2$};
\draw [fill=black] (-3.,-1.) circle (2.5pt);
\draw[color=black] (-3.05,-1.4) node {$J_3$};
\draw [fill=black] (0.,4.007796456697401) circle (2.5pt);
\draw[color=black] (0.,4.3) node {$r_1,1$};
\draw [fill=black] (0.,3.) circle (2.5pt);
\draw[color=black] (0.,3.3) node {$r_1,2$};
\draw [fill=black] (0.,2.) circle (2.5pt);
\draw[color=black] (0.,2.3) node {$r_1,3$};
\draw [fill=black] (0.,-1.) circle (2.5pt);
\draw[color=black] (0.,-1.3) node {$r_2,1$};
\draw [fill=black] (0.,-2.) circle (2.5pt);
\draw[color=black] (0.,-2.3) node {$r_2,2$};
\draw [fill=black] (0.,1.) circle (2.5pt);
\draw[color=black] (0.,1.3) node {$r_1,4$};
\draw [fill=black] (0.,-3.) circle (2.5pt);
\draw[color=black] (0.,-3.3) node {$r_2,3$};
\draw [fill=black] (0.,-4.) circle (2.5pt);
\draw[color=black] (0.,-4.3) node {$r_2,4$};
\draw [fill=black] (4.2,3.) circle (2.5pt);
\draw[color=black] (4.3,3.3) node {$M_1,2$};
\draw [fill=black] (4.2,1.9905667868879942) circle (2.5pt);
\draw[color=black] (4.3,2.3) node {$M_1,3$};
\draw [fill=black] (4.2,1.) circle (2.5pt);
\draw[color=black] (4.3,1.3) node {$M_1,4$};
\draw [fill=black] (4.2,-1.) circle (2.5pt);
\draw[color=black] (4.3,-1.33) node {$M_2,1$};
\draw [fill=black] (4.2,-2.) circle (2.5pt);
\draw[color=black] (4.3,-2.33) node {$M_2,2$};
\draw [fill=black] (4.2,-3.) circle (2.5pt);
\draw[color=black] (4.3,-3.33) node {$M_2,3$};
\draw[color=black] (6,2.4) node {\color{darkgray} 1};
\draw[color=black] (6,1.8) node {\color{darkgray} 2};
\draw[color=black] (6,1.25) node {\color{darkgray} 3};
\draw[color=black] (6,0.7) node {\color{darkgray} 4};
\draw[color=black] (6,-0.7) node {\color{darkgray} 1};
\draw[color=black] (6,-1.25) node {\color{darkgray} 2};
\draw[color=black] (6,-1.8) node {\color{darkgray} 3};
\draw[color=black] (6,-2.4) node {\color{darkgray} 4};
\draw [fill=black] (1.2,4.) circle (2.5pt);
\draw[color=black] (1.1,4.35) node {$r_1,1'$};
\draw [fill=black] (1.2,3.) circle (2.5pt);
\draw[color=black] (1.1,3.35) node {$r_1,2'$};
\draw [fill=black] (1.2,2.) circle (2.5pt);
\draw[color=black] (1.1,2.35) node {$r_1,3'$};
\draw [fill=black] (1.2,1.) circle (2.5pt);
\draw[color=black] (1.1,1.35) node {$r_1,4'$};
\draw [fill=black] (1.2,-1.) circle (2.5pt);
\draw[color=black] (1.1,-1.28) node {$r_2,1'$};
\draw [fill=black] (1.2,-2.) circle (2.5pt);
\draw[color=black] (1.1,-2.28) node {$r_2,2'$};
\draw [fill=black] (1.2,-3.) circle (2.5pt);
\draw[color=black] (1.1,-3.28) node {$r_2,3'$};
\draw [fill=black] (1.2,-4.) circle (2.5pt);
\draw[color=black] (1.1,-4.28) node {$r_2,4'$};
\end{scriptsize}
\end{tikzpicture}
\caption{Min Cost-Flow instance for $P|\textit{partition}, p_j=1|\sum_j C_j$ with 4 jobs and 2 resources.}
\end{figure}
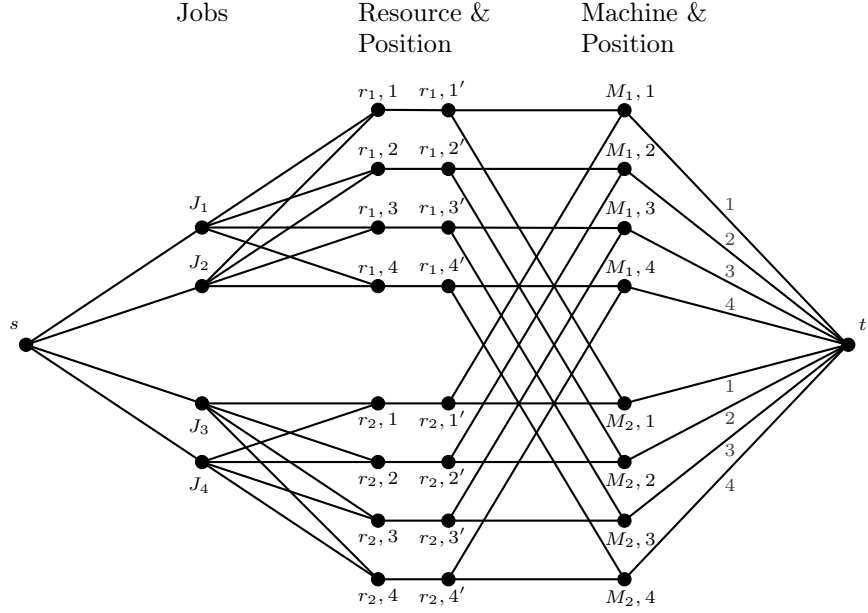
Suppose we have an instance of $P|\textit{partition}, p_j=1|\sum_j C_j$ with objective value $c$ and consider its corresponding Min-Cost Flow instance. We put one unit of flow in the network for every job in the schedule of $P|\textit{partition}, p_j=1|\sum_j C_j$ corresponding to the job, its position, which machine and which resource used. For example, for a job $j$ on position $p$ machine $M$ using resource $r$, we put one unit of flow from $s$ to $t$ through nodes $v_{J}$, $v_{r,p}$, $v'_{r,p}$ and $M,p$. Since no jobs share the same machine or resource at a given position, every node $v \in V\setminus \{s,t\} $ will have at most one unit of flow going in and going out. Thus, we have a feasible flow. Furthermore, we only put flow on edges from $v_{m,p} \in V_{M,\text{pos}}$ to $t$, if and only if we also have a job with corresponding machine $m$ and position $p$, thus we have a flow of cost $c$. The reverse is also true by this construction and hence we have an objective value of $c$ for $P|\textit{partition}, p_j=1|\sum_j C_j$ if and only if we have an objective value of $c$ for its corresponding Min-Cost Flow instance.
\end{proof}

Note that, the above proof can easily be adjusted to the problem where one has more than one unit of a resource available, by setting the capacities of the arcs between $V_{\text{res},\text{pos}}$ and $V'_{\text{res},\text{pos}}$ equal to the amount of resources one has. Furthermore, notice that one could put the costs on a different set of edges. In this way, one can also show that $P|\textit{partition}, p_j=1|\sum_j w_j C_j$ is polynomially solvable by setting the cost of the edges between $v_J$ and $v_{r,p}$ equal to $w_j$ times the position.\\

Since we know that $P|\textit{partition}, p_j=1|\sum_j C_j$ is polynomially solvable, one might wonder what happens when the processing time of each job is bounded, i.e. $1 \leq p_j \leq c$, where $c$ is a constant. A simple Shrinking algorithm would be to create an instance of $P|\textit{partition}, p_j=1|\sum_j C_j$  by setting all processing times in our original problem to one. We then solve this problem using the construction in Theorem \ref{res_p1} to obtain an optimal schedule $S^{\textsc{Opt}}_{p_j=1}$. We can then construct a feasible schedule for $P|\textit{partition}, 1 \leq p_j\leq c |\sum_j C_j$ in the following way: All jobs in $S^{\textsc{Opt}}_{p_j=1}$ start at a given integer $s_j \in \mathbb{Z}^+$, since $p_j=1$. We can create a feasible solution $S_{\textsc{Alg}}$ from  $S^{\textsc{Opt}}_{p_j=1}$ by setting the starting time for all jobs to $c s_j$. In this way, in every time interval $c i\leq t \leq c (i+1)$ with $i \in \mathbb{Z}^+$, every machine will only work on one product. There also will not be any resource conflicts, since there were none in $S^{\textsc{Opt}}_{p_j=1}$ in the corresponding interval $i\leq t \leq (i+1)$.

\begin{proposition}
The Shrinking algorithm gives a $c$-approximation for $P|\textit{partition}, \allowbreak 1 \leq p_j\leq c |\sum_j C_j$. 
\end{proposition}
\begin{proof}
Let $\textsc{Opt}$ denote the optimal value for $P|\textit{partition}, 1 \leq p_j\leq c |\sum_j C_j$ and $S^{\textsc{Opt}}$ the optimal schedule and $\textsc{Opt}_{p_j=1}$ denote the optimal value for $P|\textit{partition}, p_j=1 |\sum_j C_j$. The Schrinking algorithm gives a feasible solution of cost $c \textsc{Opt}_{p_j=1}$, so it suffices to show that $\textsc{Opt} \geq \textsc{Opt}_{p_j=1}$. Using $S^{\textsc{Opt}}$, we can find a schedule $S_{p_j=1}$ for our constructed instance of $P|\textit{partition}, p_j=1|\sum_j C_j$. This is done by scheduling all jobs 1 time unit before there completion time in $S^{\textsc{Opt}}$. Thus all completion times in $S_{p_j=1}$ will be the same as in $S^{\textsc{Opt}}$. Furthermore, since $1 \leq p_j\leq c$ in $S^{\textsc{Opt}}$, there will be no resource conflict in $S_{p_j=1}$, since there were none in $S^{\textsc{Opt}}$. 
\end{proof} 
Note that one can remove all idle time in $S_{\textsc{Alg}}$ by using the untangling operation described in the proof of Lemma \ref{idletime}.

\section{Shortest processing time first}
The shortest processing time first (SPT) rule is optimal for a few scheduling problems, one of which is $P| |\sum_j C_j$. In this section, we will look at how well it performs for $P|\textit{partition}|\sum_j C_j$. Before we can do this however we need to adjust the rule slightly to cope with the resources. 

\begin{definition}
The SPT-\textit{available} rule schedules the jobs according to a list. This list contains all jobs ordered for shortest to largest processing time. At any point in time, when a machine is available for processing. The rule selects the first job in the list for which the resource is not in use. It then removes the job from the list. If multiple machines are available at time $t$ and a job is selected of which the resource was not available just before time $t$, the algorithm will put this job on the machine that was previously using this resource. Otherwise the rule will choose an arbitrary available machine. No job is added to a machine that is available if the resource is in use of all jobs on the list.
\end{definition}

Because of the way we defined the SPT-\textit{available} rule, jobs that share the same resource and that are processed one after another will be scheduled on the same machine. In other words, the SPT-\textit{available} rule produces a tight schedule. This schedule also has no idle times, since if at time $t$ a job $j$ of resource $r$ is scheduled on machine $m$ which would create an idle time, then it could not be scheduled on that machine earlier due to some other job $j'$ using the same resource on some machine $m'$. $j$ can only be scheduled at time $t$ since $j'$ just finished. But the rule states that job $j$ then has a preference for machine $m'$ over $m$ (creating a tight schedule).\\
The SPT-rule is optimal for $P| |\sum_j C_j$. Hence, one might wonder whether this is also the case with the SPT-\textit{available} rule for $P|\textit{partition}|\sum_j C_j$. Example \ref{SPTnotopt} shows that this is not the case. 

\begin{example}{(SPT-\textit{available} not optimal)}\label{SPTnotopt}
Consider 2 machines with 12 jobs that use 3 resources. We divide the jobs $J$ into 3 groups depending on the resource used, $J= J_1 \cup J_2 \cup J_3$. We have 4 jobs in $J_1 = \{j_1,\ldots,j_4\}$ with $p_1=\ldots=p_4=1$ using the first resource. We have another 4 jobs in $J_2 = \{j_5,\ldots,j_8\}$ with $p_5=\ldots=p_8=1$ using the second resource. Lastly, we have 4 jobs in $J_3 = \{j_9,\ldots,j_{12} \}$ with $p_9=\ldots=p_{12}=1+\varepsilon$ with $\varepsilon> 0$ using the third resource. \\

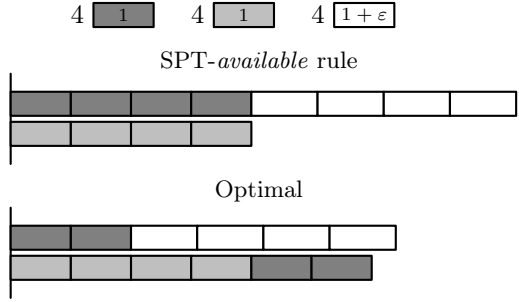
\begin{figure}[ht!]
\centering
\begin{tikzpicture}[line cap=round,line join=round,>=triangle 45,x=0.4cm,y=0.4cm]
\clip(-3,1.7) rectangle (14,3.3);

\draw[color=black] (0,2.5) node [left] {4};
\fill[fill =gray,line width=0.8pt,draw = black] (0.,2.9) rectangle (2.,2.1);
\draw[color=black] (1,2.5) node {\scriptsize $1$};

\draw[color=black] (4,2.5) node [left] {4};
\fill[fill =lightgray,line width=0.8pt,draw = black] (4.,2.9) rectangle (6.,2.1);
\draw[color=black] (5,2.5) node {\scriptsize $1$};

\draw[color=black] (8,2.5) node [left] {4};
\fill[fill =white,line width=0.8pt,draw = black] (8.,2.9) rectangle (10.,2.1);
\draw[color=black] (9,2.5) node {\scriptsize $1+\varepsilon$};
\end{tikzpicture}

{\small SPT-\textit{available} rule\\}

\begin{tikzpicture}[line cap=round,line join=round,>=triangle 45,x=0.4cm,y=0.4cm]
\clip(-1,0.3475062180393946) rectangle (17.5,3.6);

\draw [line width=0.8pt] (0.,0.5)-- (0.,3.5);

\fill[fill =gray,line width=0.8pt,draw = black] (0.,2.9) rectangle (2.,2.1);
\fill[fill =gray,line width=0.8pt,draw = black](2.,2.9)  rectangle (4.,2.1);
\fill[fill =gray,line width=0.8pt,draw = black] (4.,2.9)  rectangle (6.,2.1);
\fill[fill =gray,line width=0.8pt,draw = black] (6.,2.9) rectangle (8.,2.1);

\fill[fill =white,line width=0.8pt,draw = black] (8.,2.9) rectangle (10.2,2.1);
\fill[fill =white,line width=0.8pt,draw = black] (10.2,2.9)  rectangle (12.4,2.1);
\fill[fill =white,line width=0.8pt,draw = black] (12.4,2.9)  rectangle (14.6,2.1);
\fill[fill =white,line width=0.8pt,draw = black] (14.6,2.9)  rectangle (16.8,2.1);

\fill[fill =lightgray,line width=0.8pt,draw = black] (0.,1.1)  rectangle (2.,1.9);
\fill[fill =lightgray,line width=0.8pt,draw = black] (2.,1.1)   rectangle (4.,1.9);
\fill[fill =lightgray,line width=0.8pt,draw = black] (4.,1.1)   rectangle (6.,1.9);
\fill[fill =lightgray,line width=0.8pt,draw = black] (6.,1.1)   rectangle (8.,1.9);
\end{tikzpicture}

{\small Optimal\\}
\begin{tikzpicture}[line cap=round,line join=round,>=triangle 45,x=0.4cm,y=0.4cm]
\clip(-1,0.3475062180393946) rectangle (17.5,3.6);

\draw [line width=0.8pt] (0.,0.5)-- (0.,3.5);

\fill[fill =gray,line width=0.8pt,draw = black] (0.,2.9)  rectangle (2.,2.1);
\fill[fill =gray,line width=0.8pt,draw = black] (2.,2.9)   rectangle (4.,2.1);

\fill[fill =white,line width=0.8pt,draw = black] (4.,2.9)   rectangle  (6.2,2.1);
\fill[fill =white,line width=0.8pt,draw = black] (6.2,2.9)  rectangle  (8.4,2.1);
\fill[fill =white,line width=0.8pt,draw = black] (8.4,2.9)   rectangle (10.6,2.1);
\fill[fill =white,line width=0.8pt,draw = black] (10.6,2.9)   rectangle  (12.8,2.1);

\fill[fill =lightgray,line width=0.8pt,draw = black] (0.,1.1)   rectangle (2.,1.9);
\fill[fill =lightgray,line width=0.8pt,draw = black] (2.,1.1)  rectangle  (4.,1.9);
\fill[fill =lightgray,line width=0.8pt,draw = black] (4.,1.1) rectangle  (6.,1.9);
\fill[fill =lightgray,line width=0.8pt,draw = black] (6.,1.1) rectangle (8.,1.9);

\fill[fill =gray,line width=0.8pt,draw = black] (8.,1.1) rectangle (10.,1.9);
\fill[fill =gray,line width=0.8pt,draw = black] (10.,1.1) rectangle  (12.,1.9);

\end{tikzpicture}
\caption{Optimal and SPT-\textit{available} schedule for Example \ref{SPTnotopt}}
\end{figure}

The SPT-\textit{available} rule will schedule the jobs in $J_1 \cup J_2$ first on the two machines and will then at time 4 schedule the jobs in $J_3$ on one machine one after another. This will result in an objective value of $46+10\varepsilon$. An optimal schedule would be to schedule first two jobs from $J_1$ on the first machine and then all jobs from $J_3$. On the second machine all jobs from $J_2$ are scheduled first and then the last two jobs from $J_1$. This results in a schedule with objective value $42+10\varepsilon$. Hence SPT-\textit{available} is not optimal. 
\end{example}

The SPT-\textit{available} rule is not optimal for $P|\textit{partition}|\sum_j C_j$, but it might give a good approximation. Example \ref{SPTnotopt} gives a type of instance that is hard to tackle for the SPT-\textit{available} rule. We generalize this example to find a lower bound on the approximation factor.

\begin{lemma}
The SPT-\textit{available} rule does not give an $\alpha$-approximation for \\
$P|\textit{partition}|\sum_j C_j$ with $\alpha< \tfrac{4}{3}$.
\end{lemma}
\begin{proof}
Consider the instance $\mathcal{I}$ with 3 machines and job set $J = J_A \cup J_B$. $J_A$ consists of $3c$ jobs of length $1$, that all use their own resource, with $c \in \mathbb{Z}^+$ even. The set $J_B$ consist of $3c$ jobs of length $1+\varepsilon$ with $\varepsilon > 0$, that all share the same resource, i.e., $r_j=r_{j'} , \forall j,j'\in J_B$. \\ 
The SPT-\textit{available} rule will first schedule $c$ jobs from $J_A$ on every machine. Then, it will schedule at time $c$ all jobs from $J_B$ on one machine. Let $\texttt{ALG}_{\mathcal{I}}$ be the objective value of the SPT-\textit{available} rule for instance $\mathcal{I}$. Then,
\begin{eqnarray}
\texttt{ALG}_{\mathcal{I}} & = & \tfrac{3}{2} c (c+1) + 3 c^2 + \tfrac{3}{2} c (3c+1) (1+ \varepsilon) \nonumber \\
		& = & 9 c^2 + 3 c + \tfrac{1}{2} (9c^2+3c)\varepsilon \nonumber
\end{eqnarray}
An optimal schedule would schedule the $J_b$ jobs on a single machine and schedule the $J_A$ jobs evenly on the remaining two machine. Let $\texttt{OPT}_{\mathcal{I}}$ be the objective value of the optimal schedule for instance $\mathcal{I}$. Then,
\begin{eqnarray}
\texttt{OPT}_{\mathcal{I}} & = &  \tfrac{3}{2} c (\tfrac{3}{2}c+1) + \tfrac{3}{2} c (3c+1) (1+ \varepsilon) \nonumber \\
		& = & \tfrac{27}{4} c^2 + 3 c + \tfrac{1}{2} (9c^2+3c)\varepsilon \nonumber 
\end{eqnarray}
Thus,
\begin{eqnarray}
\frac{\texttt{ALG}_{\mathcal{I}}}{\texttt{OPT}_{\mathcal{I}}} & \geq &  \lim_{c \rightarrow \infty}  \lim_{\varepsilon \rightarrow 0} \frac{9 c^2 + 3 c + \tfrac{1}{2} (9c^2+3c)\varepsilon}{\tfrac{27}{4} c^2 + 3 c + \tfrac{1}{2} (9c^2+3c)\varepsilon} \nonumber \\
		& = & \lim_{c \rightarrow \infty}  \frac{9 c^2 + 3 c }{\tfrac{27}{4} c^2 + 3 c} = \frac{4}{3} \nonumber
\end{eqnarray}

\end{proof}

We will proceed by giving an upper bound to the approximation ratio by using an approach similar to the approach used by \citet{chekuri2001approximation} to show a 2-approximation for the problem of minimizing weighted completion time on $m$ parallel machines with in-tree precedence constraints. We begin by defining the minimum completion time of every job, based on the fact that by Theorem \ref{SPTorderOPT} all jobs sharing the same resource must be processed in SPT-order. Without loss of generality we can assume that the jobs are ordered according to their processing times, $j_1\prec j_2 \prec \cdots \prec j_n$, breaking ties arbitrarily.

\begin{definition}
The \emph{minimum completion time} $k_j$ for each job $j$ is given by 
$$k_j = p_j + \sum_{j'|r_{j'}=r_{j} \text{ and } p_{j'} \prec p_j} p_{j'}.$$ 
\end{definition}

Define $\text{OPT}^m$ as the optimal value for a given instance of jobs on $m$ machines and let $\text{OPT}^m_{\text{res}}$ be the optimal value for the instance of jobs on $m$ machines with \textit{partition} constraints. Clearly, $\text{OPT}^1 = \text{OPT}^1_{\text{res}}$ for each instance since the additional constraints do not interfere with the optimal schedule for one machine. Notice that the optimal schedule for one machine and for parallel machines is created by the SPT rule \cite{conway1967}. Let $C^1_j$ denote the completion time of job $j$ in an optimal schedule using one machine (with or without \textit{partition} constraints) and $C^m_j$ the completion time of job $j$ in an optimal schedule for $m$ machines without \textit{partition} constraints. 

\begin{lemma}\label{opt1optm}
$\frac{1}{m}\text{OPT}^1 \le\text{OPT}^m \le \text{OPT}^m_{\text{res}}$ for each instance.
\end{lemma}
\begin{proof}
Clearly, the last inequality holds, as an optimal schedule for the problem with constraints is always a feasible solution to the problem without the partition constraints. Hence, we only have to show the first inequality. \\
Sort the jobs from small to large processing times. Let $j$ be arbitrary but fixed job and let $P_j = \sum_{i=1}^j p_i$ be the sum of all processing times of the jobs that have a smaller or equal processing time. Clearly $C_j^m \geq \tfrac{1}{m} P_j$, since $\tfrac{1}{m} P_j$ is the earliest time $j$ can finish. We also have that $P_j = C_j^1$, since the SPT-rule is optimal. Thus, for every job $j$ we also have that $\tfrac{1}{m} C_j^1 \leq C_j^m$ from which the first inequality follows.  
\end{proof}

We can now prove the upper bound for the SPT-\textit{available} rule. 

\begin{theorem} \label{2approx}
The SPT-\textit{available} rule gives a $\left(2-\tfrac{1}{m}\right)$-approximation for\\ $P|\textit{partition}|\sum_j C_j$.
\end{theorem}
\begin{proof}
Let $C_j^G$ be the completion time of job $j$ in a schedule created by the SPT-\textit{available} rule. We begin by proving by induction that 

\begin{equation}\label{2approx_eq1}
C_j^G \le \left(1-\frac{1}{m}\right)k_j + \frac{1}{m}C_j^1, \quad \forall j \in J.
\end{equation}

The jobs that are the first to be scheduled on a machine are also the first of their resource. Therefore, 
\begin{align*}
C_j^G &= p_j \\
&= \left( 1 - \frac{1}{m}\right) p_j + \frac{1}{m} p_j \\
&= \left(1-\frac{1}{m}\right)k_j + \frac{1}{m} C_j^1\\
\end{align*}
in that case.\\
Now assume that Equation (\ref{2approx_eq1}) holds for any job $j' \prec j$. Then, in particular, it also holds for the job $j'$ that is scheduled right before job $j$ on the same machine. We distinguish the following two cases:

\begin{itemize}
\item \textit{$j'$ and $j$ use the same resource..} $j$ is scheduled right after $j'$ and thus
\begin{align*}
C_j^G &= C_{j'}^G + p_j \\
&\le \left(1-\frac{1}{m}\right)k_{j'} + \frac{1}{m}C_{j'}^1 + p_j  &\text{(by induction)}\\
&= \left(1-\frac{1}{m}\right)(k_{j'}+p_j) + \frac{1}{m}(C_{j'}^1 +p_j) \\
&\le \left(1-\frac{1}{m}\right)k_j + \frac{1}{m}C_{j}^1. \\
\end{align*}

\item \textit{$j'$ and $j$ use different resources.} This implies that $j$ was scheduled at the first possible free machine and not at a later possibility because of the resource constraints. Define job $j''$ as the job that is right before $j$ in the schedule for one machine, i.e. the last $j''$ in the ordering such that $j'' \prec j$. For the starting time of job $j$ (equal to $C^G_{j'}$ ) it then holds that $C^G_{j'} \le \frac{1}{m} C^1_{j''}$. Using this we see that: 

\begin{align*}
C_j^G &= C_{j'}^G + p_j \\
&\le \frac{1}{m}C_{j''}^1 + p_j \\
&= \left(1-\frac{1}{m}\right)p_j + \frac{1}{m}(C_{j''}^1 +p_j) \\
&\le \left(1-\frac{1}{m}\right)k_j + \frac{1}{m}C_{j}^1 &(\text{since } p_j \le k_j  \,\forall j). \\
\end{align*}

\end{itemize}
Hence, we can conclude that (\ref{2approx_eq1}) holds for all jobs $j \in J$.\\

Note that $\sum_jk_j \le OPT^m_{res}$, since $k_j \le C_j$ in any feasible schedule for $P|\textit{partition}|\sum_jC_j$ by the definition of the minimal completion time $k_j$. 
Using equation (\ref{2approx_eq1}) we get:


\begin{align*}
\sum_j C_j^G &\le \sum_j \left(1-\frac{1}{m}\right)k_j + \sum_j \frac{1}{m}C_j^1 & \text{(using (\ref{2approx_eq1}))}\\
&= \left(1-\frac{1}{m}\right) \sum_j  k_j + \frac{1}{m} \sum_j C_j^1\\
&\le \left(1-\frac{1}{m}\right)OPT^m_{\text{res}} + OPT^m_{\text{res}} &\text{(Lemma \ref{opt1optm})} \\
&\le \left(2-\frac{1}{m}\right)OPT^m_{\text{res}}
\end{align*}

\end{proof}

\section{Machine subset constraints}
Since we do not know the complexity of $P|\textit{partition}|\sum_j C_j$, it is interesting to look at related problems. We will look at several of these related problems. We begin by considering the problem where jobs that share the same resource can only be processed on a subset of the machines.\\ 

We can add processing set restrictions by adding $\mathcal{M}_j$ to the $\beta$ field of a scheduling problem, as found in \cite{leung2008scheduling}. This means that for each job $j$, there is a set $\mathcal{M}_j \subseteq \{1,...,m\}$ such that $j$ can only be scheduled on machines in $\mathcal{M}_j$. Let us define a variation called processing set restrictions for resources as follows: for each resource $r\in \mathcal{R}$ there is a set $\mathcal{M}_r \subseteq \{1,...,m\}$ such that all jobs sharing resource $r$ can only be scheduled  on machines in $\mathcal{M}_r$. We denote these restrictions as $\mathcal{M}_r$ in the $\beta$ field. 

\begin{corollary}$P | \textit{partition}, \mathcal{M}_r, p_j =1 | \sum_jC_j$ is polynomially solvable.
\end{corollary}

This is a consequence of Theorem \ref{res_p1}. One could use the same algorithm, but only include edges $v'_{r,p}$ to $v_{i,p}$ if
$i \in \mathcal{M}_r$. \\

Consider the following NP-complete problem from \cite{garey1979c}.
\begin{definition}{3-PARTITION}\label{3part} Given positive integers $m$ and $b$, and a multiset of $3m$ integers $A$ with $\sum_{a \in A} a = mb$, and $b/4 \le a \le b/2$ for all $a \in A$, does there exist a partition $(A_1,...,A_m)$ of $A$ into 3 element sets such that for each $i$, $1 \le i \le m$, $\sum_{a \in A_i} a = b$?
\end{definition}

This problem is NP hard in the strong sense. Using this, we will prove the following theorem.

\begin{theorem}$P | \textit{partition}, \mathcal{M}_r| \sum_j C_j$ is NP hard in the strong sense.\label{NPhardMr} 
\end{theorem}

\begin{proof}
Assume we have an instance of 3-PARTITION. Since 3-PARTITION is NP-complete in the strong sense, we may assume that $mb$ is bounded by a polynomial in $m$, which is crucial for our proof. Define $N_c = 2mb$ and $C = 8mb$. Create an instance of $P | partition, \mathcal{M}_r| \sum_j C_j$ with $2m$ machines and the following jobs:

\begin{itemize}
\item for all $a\in A$ make job $a_j$, with processing time $p_{a_j} =a$, a unique resource $r(a_j)$ and $\mathcal{M}_{r(a_j)} = \{1,2,...,m\}$. These jobs represent the integers that should be partitioned over the first $m$ machines.
\item for all $  1\le i \le m$ make $N_c$ jobs called `$C$'-jobs with processing time $C$, resource $i$ and $\mathcal{M}_{i} = \{i,m+i\}$, so for each $i$, there are $N_c$ jobs with length $C$, all sharing the same resource, that can only be scheduled on machines $i$ and $m+i$.
\item for all $  1\le i \le m$ make job $r_i$, also called a release date job with processing times $p_{r_i} =b$ and resource $i$, so it shares its resource with a sequence of `$C$'-jobs and can only be scheduled on machines $i$ and $m+i$. 
\item for all $  1\le i \le m$ make job $D_i$, also called a `$D$'-job, with processing time $p_{D_i} = N_C^2C$, resource $r(D_i)$ and $\mathcal{M}_{r(D_i)} = \{m+i\}$, so its resource is unique and can only be scheduled on machine $m+i$.  
\end{itemize}

Define $Z^+ = mb + m \left(N_C b + (C+N_C  C) \frac{N_C}{2} \right) + m(b + N_c^2C) + 2mb$. We will show that the optimal schedule for the scheduling problem has objective value $Z^* \le Z^+$ if and only if the 3-PARTITION instance is a yes-instance. \\

Assume that the 3-PARTITION instance is a yes-instance, then the following schedule is a feasible solution: We can find $A_i$ with $|A_i| =3$ s.t. $\sum_{a \in A_i}^m a = b$ for all $i$. Schedule each of these $A_i$ at the beginning of one of the first $m$ machines. Process the jobs from small to big in their processing times per machine. Start the release date jobs $r_i$ at machines $m+i$ at $t=0$. Process the `$C$'-jobs from $t=b$ and onwards at the first $m$ machines. Start each `$D$' job $D_i$ at machine $m+i$ at $t=b$. For a visualization of this schedule see Figure \ref{Feasible Solution}. The objective value of such a solution is equal to 

$$Z_{feas} = \underbrace{\vphantom{ \left(\frac{N_C}{2}\right)} m\cdot b}_{\text{release date jobs}} + \underbrace{m \left(N_C \cdot b + (C+N_C \cdot C) \frac{N_C}{2} \right)}_{`C'\text{ jobs}} + \underbrace{\vphantom{ \left(\frac{N_C}{2}\right)} m\left(b+N_C^2\cdot C\right)}_{`D' \text{ jobs}} + \underbrace{\vphantom{ \left(\frac{N_C}{2}\right)} Z_A}_{a_j \text{ jobs}}, $$ 

with $\frac{7}{4} mb \le Z_A \le 2 mb$ since the following: $\frac{b}{4} \le a_j \le \frac{b}{2}$, and the $a_j$ jobs are sorted from small to big in their processing times per machine, so the worst case scenario is if $A_i$ only has jobs of processing times $\frac{b}{3}$ and the best case scenario is if $A_i$ has jobs of processing times $\frac{b}{4}, \frac{b}{4}$ and $\frac{b}{2}$. Since $Z_{feas} \le Z^+$, we can conclude that $Z^*\le Z^+$.\\


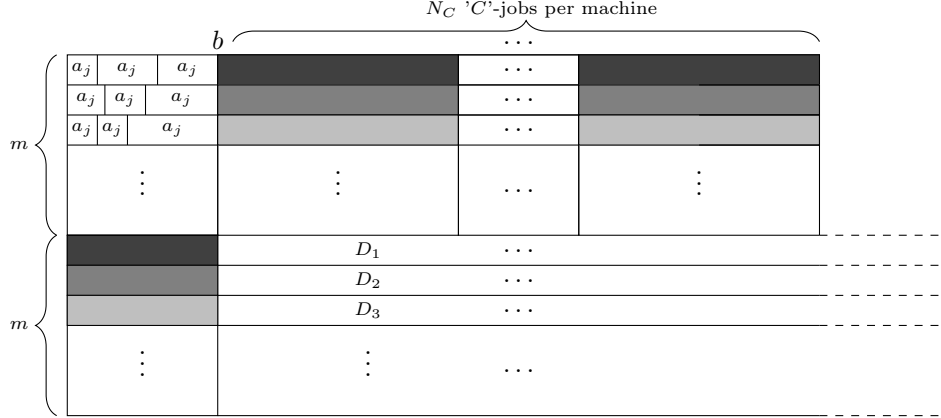
\begin{figure}[h!]
\centering
\begin{tikzpicture}[scale = 0.4]

\draw (0,0) rectangle (5,6);

\fill[fill=darkgray, draw = black] (0,5) rectangle (5,6);
\fill[fill=gray,  draw = black] (0,4) rectangle (5,5);
\fill[fill=lightgray,  draw = black] (0,3) rectangle (5,4);
\node at (2.5,2){$\vdots$};

\draw (0,6) rectangle (5,12);
\draw (0,11) -- (5,11);
\draw (0,10) -- (5,10);
\draw (0,9) -- (5,9);

\draw (1,11) -- (1,12);
\node at (0.5,11.5) { \scriptsize $a_j$};
\draw (3,11) -- (3,12);
\node at (2,11.5) { \scriptsize $a_j$};
\node at (4,11.5) { \scriptsize $a_j$};

\draw (1.25,10) -- (1.25,11);
\node at (0.65,10.5) { \scriptsize $a_j$};
\draw (2.6,10) -- (2.6,11);
\node at (1.925,10.5) { \scriptsize $a_j$};
\node at (3.8,10.5) { \scriptsize $a_j$};

\draw (1,9) -- (1,10);
\node at (0.5,9.5) { \scriptsize $a_j$};
\draw (2,9) -- (2,10);
\node at (1.5,9.5) { \scriptsize $a_j$};
\node at (3.5,9.5) { \scriptsize $a_j$};

\node at (2.5,8) {$\vdots$};

\draw (5,6) rectangle (13,12);
\fill[fill =darkgray, draw = black] (5,11) rectangle (13,12);
\fill[fill=gray,  draw = black] (5,10) rectangle (13,11);
\fill[fill=lightgray,  draw = black] (5,9) rectangle (13,10);
\node at (9,8) {$\vdots$};

\node at (5,12.5) {$b$};

\draw (13,11) rectangle (17,12);
\draw (13,10) rectangle (17,11);
\draw (13,9) rectangle (17,10);
\node at (15,11.5) {$\hdots$};
\node at (15,10.5) {$\hdots$};
\node at (15,9.5) {$\hdots$};
\node at (15,7.5) {$\hdots$};
\node at (15,5.5) {$\hdots$};
\node at (15,4.5) {$\hdots$};
\node at (15,3.5) {$\hdots$};
\node at (15,1.5) {$\hdots$};
\node at (15,12.4) {$\hdots$};

\draw (13,6) rectangle (17,12);
\draw (25,12) -- (17,12) -- (17,6) -- (25,6);

\fill[fill =darkgray,draw = black] (17,11) rectangle (25,12);
\fill[fill=gray,  draw = black] (17,10) rectangle (25,11);
\fill[fill=lightgray,  draw = black] (17,9) rectangle (25,10);

\draw (21,11) -- (25,11);
\draw (21,10) -- (25,10);
\draw (21,9) -- (25,9);
\draw (25,11) -- (25,6);

\draw [dashed] (25,6) -- (29,6);
\draw (5,5) -- (25,5);
\draw [dashed] (25,5) -- (29,5);
\draw (5,4) -- (25,4);
\draw [dashed] (25,4) -- (29,4);
\draw (5,3) -- (25,3);
\draw [dashed] (25,3) -- (29,3);
\node at (10,5.5) {\scriptsize $D_1$};
\node at (10,4.5) {\scriptsize $D_2$};
\node at (10,3.5) {\scriptsize $D_3$};
\node at (10,2) {$\vdots$};
\draw (5,0) -- (25,0);
\draw [dashed] (25,0) -- (29,0);

\node at (21,8) {$\vdots$};

\draw[decorate,decoration={brace,raise=5pt,amplitude=8pt}]
    (5.5,12)--(25,12) ;
\node at (15.75,13.5) { \scriptsize $N_C$ '$C$'-jobs per machine};
\draw[decorate,decoration={brace,raise=4pt,amplitude=8pt}]
    (0,0)--(0,6) ;
\node at (-1.6,3) {\scriptsize $m$};
\node at (-1.6,9) {\scriptsize $m$};
\draw[decorate,decoration={brace,raise=4pt,amplitude=8pt}]
    (0,6)--(0,12) ;


\end{tikzpicture}
\caption{Feasible solution in the case of a yes-instance.}
\label{Feasible Solution}
\end{figure}

We will show that if the 3-PARTITION instance is a no-instance, the optimal schedule has an objective value $Z^* > Z^+$. Let $Z^* = Z^*_r + Z^*_C + Z^*_D +Z^*_a$ with $Z^*_r$,$Z^*_C$,$Z^*_D$ and $Z^*_a$ the sum of the completion times of the release date jobs, `$C$'-jobs, `$D$'-jobs and $a_j$-jobs respectively. Notice that  $mb$ is a lower bound on $Z^*_r$ since the jobs cannot start before $t=0$. In the same way, $mN_C^2C$ is a lower bound on $Z^*_D$. A lower bound on $Z^*_a$ is $\frac{7}{4} mb$, this is because if only the $a_j$-jobs were to be scheduled on $m$ machines, SPT-order would be optimal and would have 3 jobs on every machine. Suppose not, then there is a machine $i_1$ with 4 jobs or more. Then there is another machine $i_2$ with 2 or less jobs. Moving the first job $j$ on machine $i_1$ to machine $i_2$ would result in at least $3$ jobs finishing $p_j$ earlier and at most $2$ jobs finishing $p_j$ later at machine $i_2$. This leads to a contradiction that SPT is optimal. So we may assume each machine has exactly 3 jobs for finding the lower bound of $Z^*_a$. Then $ \sum_{i=1}^m (3a_{i1} + 2a_{i2} + a_{i3})$ is minimal if one chooses all $a_{i1}$ and $a_{i2}$ to be $\frac{b}{4}$, i.e. as small as possible. This implies $a_{i3} = \frac{b}{2}$, as $\sum_{a\in A} a = mb$ and $\frac{b}{4}\le a \le \frac{b}{2}$ for all $a \in A$. This leads to a total completion time and therefore a lower bound of $m \left(\frac{b}{4} + \frac{b}{2} + b\right) = \frac{7}{4} mb$ for $Z^*_a$. 

If we have a no-instance, we argue that in the optimal schedule $\exists a_j$ with completion time bigger than $b$. Assume not, then all first $m$ machines are processing $a_j$ jobs until $b$, since $\sum_{a \in A} a = mb$ . If there is a machine processing more than three $a_j$ jobs, it would have to be four $a_j$ jobs of length $\frac{b}{4}$, so that the $a_j$ jobs are all finished before or at $b$. But then, there is also a machine processing only two $a_j$ jobs of length $\frac{b}{2}$, otherwise another machine would have to finish its $a_j$ jobs after $b$. Switching a $\frac{b}{2}$ job with two $\frac{b}{4}$ jobs would then result in a smaller objective value. Hence, all machines are processing exactly three $a_j$ jobs with $\sum_{a \in A_i} a = b$ for all $1\le i\le m$. Then we would find a partition, leading to a contradiction. So there must be an $a_j$ finishing after $b$. Since $a_j \in \mathbb{N}$ for all $j$, we can find an $a_j$ with completion time $\ge b+1$. We distinguish 4 cases: 

\begin{itemize}
\item \textit{At least one `$C$'-job is scheduled before the $r_i$ job with the same resource.} Let $N_i$ be the number of `$C$'-jobs on machines $i$ and $m+i$ scheduled before the corresponding release job $r_i$. Then $r_i$ starts at $N_iC$ or later. The lower bound on $Z^*_r$ becomes $mb + C \cdot \sum_{i=1}^m N_i$. Then $ \sum_{i=1}^m \left((N_C-N_i) b + (C+N_C  C) \frac{N_C}{2}\right)$ is a lower bound of $Z^*_C$. Using the other lower bounds for $Z^*_a$ and $Z^*_D$, we obtain the following lower bound for $Z^*$: 
$$ mb + C \cdot \sum_{i=1}^m N_i+ \sum_{i=1}^m \left((N_C-N_i) b + (C+N_C  C) \frac{N_C}{2}\right) + mN_C^2C+ \frac{7}{4} mb.$$ Then $$Z^*-Z^+ \ge  (C-b)\sum_{i=1}^m N_i - mb - \frac{1}{4}mb >0,$$
using that $C = 8mb$ and $\sum_{i=1}^m N_i >0$.

\item \textit{All `$C$'-job are scheduled after the $r_i$ job with the same resource, but at least one `$C$'-job is scheduled on one of the last $m$ machines.} We split this up into two subcases: 
\begin{itemize}
\item At least one such `$C$'-job is scheduled \emph{before} a `$D$'-job on the same machine. We know all `$C$'-jobs start after $b$, hence $Z^*_C \ge m(N_Cb + (N_CC+C)\frac{N_C}{2})$. Then $Z^*_D \ge mN_C^2C + b + C$ since one `$D$'-job starts after $b+C$. Using the other lower bounds we get 
$$Z^* - Z^+ \ge b+C - mb - \frac{1}{4}mb > 0, $$
using that $C = 8mb$.
\item At least one such `$C$'-job is scheduled \emph{after} a `$D$'-job on the same machine. Let $i$ be the resource of one such `$C$'-job and $D_i$ the corresponding `$D$'-job. Then the `$C$'-job finishes at $N^2_CC+C$ or later, while any `$C$'-job scheduled \emph{not} after a $D$'-job would have a maximum completion time of $\sum_{j=1}^{3m} a_j + b + N_CC$, which is the sum of all processing times of jobs that could possibly be scheduled on machine $i$. Hence $Z^*_C \ge m(N_Cb + (N_CC+C)\frac{N_C}{2}) + (N_c^2C + C - (mb + b + N_CC))$. Using the lower bounds for $Z^*_r$, $Z^*_a$ and $Z^*_D$, we get
$$Z^*-Z^+ = N_C^2C+C -(mb + b + N_CC) - mb - \frac{1}{4} mb>0, $$
using that $C = 8mb$ and $N_C = 2mb$.
\end{itemize}

\item \textit{All `$C$'-job are scheduled after the $r_i$ job with the same resource, all `$C$'-jobs are scheduled on the first $m$ machines, but least one `$C$'-job is scheduled before an $a_j$ job on the same machine.} All `$C$'-jobs are scheduled after $b$, hence $Z^*_C \ge m(N_C b + (C+N_C C) \frac{N_C}{2})$. At least one $a_j$ has completion time bigger than $C +b$, while for any $a_j$ \emph{not} scheduled after an `$C$'-job has a completion time smaller or equal to $\sum_{j=1}^{3m} a_j +b = mb +b$, so $Z^*_a \ge \frac{7}{4} mb + (C + b - mb -b)$. Using the lower bounds for $Z^*_r$ and $Z^*_D$, we get
$$ Z^* - Z^+ \ge (C-mb) - mb -\frac{1}{4}mb > 0$$
using that $C = 8mb$.

\item \textit{All `$C$'-job are scheduled after the $r_i$ job with the same resource, all `$C$'-jobs are scheduled on the first $m$ machines and all `$C$'-job are scheduled after the $a_j$ jobs on the same machine.} Notice that the feasible solution in Figure \ref{Feasible Solution} is structured in a similar way. At least one machine $i$ should have an $a_j$ job with completion time at least $b+1$. So the sequence of `$C$'-jobs on machine i should start at $b+1$ or later. This means that $Z^*_C \ge m \left(N_C b + (C+N_C  C) \frac{N_C}{2} \right) + N_C$. Using the lower bounds on $Z^*_r$, $Z^*_a$ and $Z^*_D$, we get:
$$ Z^*-Z^+ \ge N_c - mb - \frac{1}{4}mb > 0, $$
using that $N_C = 2mb$.
\end{itemize}

So if the 3-PARTITION instance is a no-instance, $Z^*>Z^+$, hence the reduction is complete.
\end{proof}

We can use Theorem \ref{NPhardMr} to show that our original problem with unrelated machines instead of parallel machines is $\mathcal{NP}$-hard. This problem is actually the problem found in the lithography bays of the Europian semiconductor factories \cite{bitar2016memetic}.

\begin{corollary}$R|\textit{partition}| \sum_j C_j$ is $\mathcal{NP}$-hard in the strong sense.
\end{corollary}
\begin{proof}
We can reduce any decision variant instance $I_P$ of $P | partition, \mathcal{M}_r| \sum_j C_j$, asking whether there exists a feasible solution with total completion time smaller than $T$, to a decision variant instance $I_R$ of $R|partition| \sum_j C_j$ asking the same question. This is done by simply removing the processing set restrictions for resources and changing the processing times to:
$$p_{ij} = \begin{cases} 
p_j & \text{ if } i \in \mathcal{M}_{r(j)} \\
T& \text{ if } i \not\in \mathcal{M}_{r(j)} \\
\end{cases}$$ where $\mathcal{M}_{r(j)}$ denotes the machine restriction for $r(j)$, the resource of job $j$. Clearly, any feasible schedule for $I_P$ is also a feasible schedule for the mapped instance $I_R$ with the same total completion time. Hence if we have a yes-instance for $I_P$, we also have a yes-instance for $I_R$. However, if we have a no instance for $I_P$, all feasible solution for $I_P$ have a total completion time at least $T$. This means that all schedules for $I_R$ processing only $j$ on $i \in \mathcal{M}_{r(j)}$ for all $j$, also have total completion time at least $T$. However, any schedule processing at least one $j$ on an $i \not \in \mathcal{M}_{r(j)}$ also has a total completion time at least $T$, because of such a job $j$. Hence $I_R$ is also a no-instance. 
\end{proof}

\section{Unmovable resources}
Moving the resources can be a costly operation. Thus one might also consider the case were the resources are also fixed on a machine. We therefore consider the problem where every resource can only be used on one machine. We define \textit{unmovable} as an addition to the \textit{partition} constraint, were all jobs $j \in r^k$ have to be processed on the same machine.  

\begin{theorem} 
$P|\textit{partition},\textit{unmovable}, p_j=1|\sum_j C_j$ is $\mathcal{NP}$-hard.
\end{theorem}
\begin{proof}
We give a polynomial time reduction from the 3-Partition problem as defined in Definition \ref{3part}. 
We introduce in $P|\textit{partition},\textit{unmovable}, p_j=1|\sum_j C_j$ $m$ machines and a number of jobs equal to $n=\sum_{a \in A} a$ and a number of resources equal to $3m$, where $M$ is a large number. For every element of $a \in A$, we associate a number of jobs equal to $a$ sharing the same resource. If we have a yes-instance of 3-Partition, then, $\forall i \in \{ 1, \ldots, m\}$, we can schedule all jobs associated with $a \in A_i$ to machine $i$. All machines will then be busy processing until time $b$. This will give us an objective value $\tfrac{m}{2} b (b+1)$. If we have a no-instance of 3-Partition, we cannot distribute the jobs evenly over the machines and thus the objective value will be greater than $\tfrac{m}{2} b (b+1)$. Hence, there exists a solution to 3-Partition if and only if $P|\textit{partition},\textit{unmovable}, p_j=1|\sum_j C_j$ as constructed above has an objective value of $\tfrac{1}{2} mb (b+1)$.\\
Note that, one might have a yes-instance of $P|\textit{partition},\textit{unmovable}, p_j=1|\sum_j C_j$, where on one machine there are 4 resources being used. If this is the case, these resources all have $b/4$ associated jobs and since it is a yes-instance, there also must be a machine using only 2 resources with $b/2$ associated jobs. An easy switch of the last $b/2$ units of processing of these two machine, will turn this also in a yes-instance for 3-Partition. 
\end{proof}

\section{Two resources per job}
Because the problem was motivated by the scheduling problem found in the lithography bays of the semi-conductor industry, we are mainly interested in the case that there is only one resource per job. However, one might also wonder what happens if there is more than one resource needed per job. We will therefore continue by looking at instances with at most $q$ resources per job. We introduce \textit{partition}$(q)$ for the $\beta$ field of the scheduling problem. If \textit{partition}$(q)$ is in the $\beta$ field, there is a collection of subsets $R = \{r^1,\ldots,r^R\}$ with $r^k \subseteq J$, where every job is contained in at most $q$ subsets. If there exist an $r^k \in R$ such that $j,j' \in r^k$, $j$ and $j'$ cannot be processed at the same time. \\
The problem $P|\textit{partition}$(q)$, p_j=1|\sum_j C_j$ is a special case of $ P|\text{res}\cdots, \text{types} = \mathcal{R}, p_i<p|f$ with $f \in \{ \sum_j w_j C_j, \sum_j T_j, \sum_j U_j \}$. Here, $s$ is the number of resources and there are $\mathcal{R}$ types of jobs. A type of a job $j$ is defined as the tuple $(p_j,\mathcal{R}_1(j),\ldots,\mathcal{R}_s(j))$, where $\mathcal{R}_u(j)$ is the amount of resource $u$ required by job $j$. Note that, in our case, $|R|=\mathcal{R}=s$. \citet{brucker1996polynomial} show that it can be solved in $O(\mathcal{R}(p+s)n^{\mathcal{R}p}+ \mathcal{R}^2 pn^{\mathcal{R}(p+2)})$, resulting in the following corollary.

\begin{corollary}
$P|\textit{partition}(q), p_j=1|\sum_j C_j$ is polynomially solvable for every $q\in \mathbb{Z}$, if the number of resources, $|R|$, is bounded. 
\end{corollary}

However, we will now show that the problem becomes $\mathcal{NP}$-hard when the number of resources is not bounded, even with $q=2$. 

\begin{theorem}\label{Th_qis2}
$P|\textit{partition}(q), p_j=1|\sum_j C_j$ is $\mathcal{NP}$-hard for every $q \geq 2$, if the number of machines $m$ and resources $|R|$ are unbounded.
\end{theorem}
\begin{proof}
We will prove this by a reduction from edge coloring. In the edge coloring problem, one assigns colors (or labels) to the edges of a graph $G=(V,E)$, such that no two incident edges have the same color. Let $\Delta$ be the maximum node degree in the graph $G$, then \citet{holyer1981np} shows that it is $\mathcal{NP}$-hard to decide for an arbitrary graph $G$ whether or not it can be colored using only $\Delta$ colors. \\
We can reduce this problem to $P|\textit{partition}(2), p_j=1|\sum_j C_j$ as follows. Suppose we are given graph $G=(V,E)$ with $|E|=m$. Take the number of machines equal to $m$. Introduce a resource for every node $u\in V$, $|R|=|V|$. We also introduce a job (with $p_j=1$) for every edge $e=\{u,v\}$ and these jobs require the resources that are associated with nodes it connects (i.e. $u$ and $v$). Thus every resource will be used at most $\Delta$ times and every job uses exactly 2 resources. Lastly, we introduce $(\Delta-1) n$ dummy jobs (with $p_j=1$), that do no require a resource. Figure \ref{example2resred} shows an example of the reduction.\\

\begin{figure}[ht!]
\centering
	\subfigure[Graph $G=(V,E)$]{
    \begin{tikzpicture}[scale = 0.65]
    	\draw [fill=black] (0,0) circle (2.5pt);
		\node at (0,-0.5) {\scriptsize $1$};
        \draw [fill=black] (-2,2) circle (2.5pt);
		\node at (-2.5,2) {\scriptsize $2$};
        \draw [fill=black] (2,2) circle (2.5pt);
		\node at (2.5,2) {\scriptsize $3$};
        \draw [fill=black] (0,4) circle (2.5pt);
		\node at (0,4.5) {\scriptsize $4$};
	
		\draw (0,0.)-- (-2,2);
        \draw (0,0.)-- (0,4);
        \draw (0,0)-- (2,2);
        \draw (-2,2)-- (0,4);
        \draw (2,2)-- (0,4);
        
	\end{tikzpicture}
    } \hspace{12pt}
	\subfigure[Optimal schedule for instance the instance. Dashed lined jobs represent dummy jobs and number represent the resources ]{ 
    \begin{tikzpicture}[scale = 0.6]
    \draw [fill=white,draw = white] (9,5) circle (2.5pt);
    
    \fill[fill=lightgray,  draw = black] (0,0) rectangle (2,1);
    \fill[fill=lightgray,  draw = black] (2,0) rectangle (4,1);
    \fill[fill=lightgray,  draw = black] (4,0) rectangle (6,1);
    \fill[fill=lightgray,  draw = black] (0,1) rectangle (2,2);
    \fill[fill=lightgray,  draw = black] (2,1) rectangle (4,2);
    \fill[fill=lightgray,  draw = black] (4,1) rectangle (6,2);
    \fill[fill=lightgray,  draw = black] (0,2) rectangle (2,3);
    \fill[fill=lightgray,  draw = black] (2,2) rectangle (4,3);
    \fill[fill=lightgray,  draw = black] (4,2) rectangle (6,3);
    \fill[fill=lightgray,  draw = black] (4,3) rectangle (6,4);
	\fill[fill =white, draw = black] (0,4) rectangle (2,5);
	\node at (1,4.5) {\footnotesize $1 \ \ 2$};
    \fill[fill =white, draw = black] (2,4) rectangle (4,5);
	\node at (3,4.5) {\footnotesize $1 \ \ 3$};
    \fill[fill =white, draw = black] (4,4) rectangle (6,5);
	\node at (5,4.5) {\footnotesize $2 \ \ 3$};
    \fill[fill =white, draw = black] (0,3) rectangle (2,4);
	\node at (1,3.5) {\footnotesize $3 \ \ 4$};
    \fill[fill =white, draw = black] (2,3) rectangle (4,4);
	\node at (3,3.5) {\footnotesize $2 \ \ 4$};
	\draw[decorate,decoration={brace,raise=4pt,amplitude=8pt}]
    (0,0)--(0,5) ;
    \node at (7.1,-0.5) {\scriptsize time};
    \node at (6,-0.5) {\scriptsize $\Delta$};
    \draw[->] (0,0) -- (7.1,0);
    \node[rotate=90] at (-1.5,2.5)  {$n$ machines};
	\end{tikzpicture}    
	}
    
\caption{Example of the reduction from edge coloring to $P|\textit{partition}(2), p_j=1|\sum_j C_j$ with a graph with $\Delta =3$}
\label{example2resred}
\end{figure}
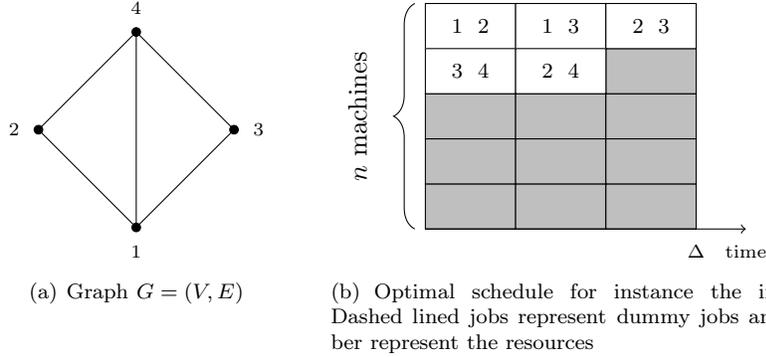

We claim that there exists an edge coloring of graph $G$ using only $\Delta$ colors if and only if the instance of $P|\textit{partition}(2), p_j=1|\sum_j C_j$ has an optimal value of $\tfrac{1}{2}\Delta(\Delta+1)m$. Suppose we are given a solution to edge coloring problem which uses $\Delta$ colors, then we can put each color on a different time slot. So all jobs associated with an edge of the first color will be put on machines in the first time slot. We fill up all unused machine time until time $\Delta$ with dummy jobs. Since jobs only share a resource if they were incident in the graph, we will not have any resource conflict and all machines will be filled with jobs until time $\Delta$, thus resulting in an objective value of $\tfrac{1}{2}\Delta(\Delta+1)m$.\\
Suppose we have a solution of the above instance of $P|\textit{partition}(2), p_j=1|\sum_j C_j$ with objective value $\tfrac{1}{2}\Delta(\Delta+1)m$. Then in each time slot we look for the jobs associated with a node and give them the same color. Since the objective value is $\tfrac{1}{2}\Delta(\Delta+1)m$ there are no jobs after time $\Delta$, hence there are only $\Delta$ colors. Since all jobs associated with incident edges share a resource, no two incident edges will share the same color and hence we have found an edge coloring using $\Delta$ colors. 
\end{proof}


\section{Conclusion}
In this paper, we considered the problem of minimizing the total completion time while scheduling jobs that each use exactly one resource, $P| \textit{partition}|\sum_j C_j$. Although the complexity of $P| \textit{partition}|\sum_j C_j$ remains unclear, we saw that similar problems such as $P|\textit{partition}, \mathcal{M}_r|\sum_j C_j$, $P| \textit{partition}(2), p_j=1|\sum_j C_j$ and $P|\textit{partition}, \textit{unmovable}, p_j=1|\sum_j C_j$ are $\mathcal{NP}$-hard. Therefore, we conjecture that $P| \textit{partition}|\sum_j C_j$ is $\mathcal{NP}$-hard as well.\\
The problem $P| \textit{partition}|\sum_j C_j$ always has an optimal solution where jobs sharing the same resource are ordered by the processing time. Such an optimal solution might even be more structured. For example, it remains open whether or not there is always an optimal solution that yields the SPT order property on each machine for all jobs on that machine. \\
We showed that the SPT-\textit{available} rule gives a $\left(2-\tfrac{1}{m}\right)$-approximation. This bound may not be tight. There is a lower bound of $\tfrac{4}{3}$ on the approximation factor. Closing this gap is another interesting open problem, as well as designing other approximation algorithms with better approximation ratios.

\bibliography{biblio}

\end{document}